\documentclass{article}
\usepackage{bbold,dsfont}
\usepackage[ruled]{algorithm2e}
\usepackage{graphicx}
\usepackage{tikz}
\usetikzlibrary{arrows}
\usepackage[utf8]{inputenc}
\usepackage{amsmath}
\usepackage{amssymb}
\usepackage{amsthm}
\usepackage{dsfont}
\usepackage{bbold}
\usepackage{mathtools}
\usepackage{graphicx}
\usepackage{tikz}
\usepackage{xcolor}
\usepackage{verbatim}
\usepackage{authblk}
\usetikzlibrary{arrows}

\newtheorem{theorem}{Theorem}
\newtheorem{proposition}{Proposition}
\newtheorem{lemma}{Lemma}
\newtheorem{example}{Example}
\newtheorem{definition}{Definition}[section] 
\newtheorem{corollary}{Corollary}
\newtheorem{remark}{Remark}
\newtheorem{problem}{Problem}
\newtheorem{assumption}{Assumption}

\title{Optimal intervention in traffic networks}
\author[1]{Leonardo Cianfanelli} 
\author[1,2]{Giacomo Como} 
\author[3]{Asuman Ozdaglar} 
\author[4]{Francesca Parise}
\affil[1]{\footnotesize Dipartimento di Scienze Matematiche, Politecnico di Torino}
\affil[2]{\footnotesize Department of Automatic Control, Lund University, BOX 118, SE-22100, Lund, Sweden}
\affil[3]{\footnotesize Department of Electrical Engineering and Computer Science, MIT}
\affil[4]{\footnotesize Department of Electrical and Computer Engineering, Cornell University}
\date{\footnotesize \EMAIL{leonardo.cianfanelli@polito.it}, \EMAIL{giacomo.como@polito.it},
	\EMAIL{asuman@mit.edu},
	\EMAIL{fp264@cornell.edu}\\
	\URL{https://sites.google.com/view/leonardo-cianfanelli/}, \URL{https://staff.polito.it/giacomo.como/},  \URL{https://asu.mit.edu/},
	\URL{https://sites.coecis.cornell.edu/parise/}}

\usepackage[colorlinks=true,breaklinks=true,bookmarks=true,urlcolor=blue,
citecolor=blue,linkcolor=blue,bookmarksopen=false,draft=false]{hyperref}

\def\EMAIL#1{\href{mailto:#1}{#1}}
\def\URL#1{\href{#1}{#1}}         

\DeclareMathOperator*{\argmax}{arg\,max}

\newcommand{\mc}{\mathcal}
\newcommand{\mb}{\mathbf}

\newcommand{\dd}{\mathrm{d}}
\newcommand{\oo}{\mathrm{o}}

\newcommand{\ov}{\overline}
\newcommand{\tv}{\tilde{v}}
\newcommand{\ty}{\tilde{y}}
\newcommand{\control}{u}
\newcommand{\Control}{\mathcal{U}}

\begin{document}
\maketitle

\begin{abstract}
We study a network design problem (NDP) where the planner aims at selecting the optimal single-link intervention on a transportation network to minimize the travel time under Wardrop equilibrium flows. Our first result is that, if the delay functions are affine and the support of the equilibrium is not modified with interventions, the NDP may be formulated in terms of electrical quantities computed on a related resistor network. In particular, we show that the travel time variation corresponding to an intervention on a given link depends on the effective resistance between the endpoints of the link.
We suggest an approach to approximate such an effective resistance by performing only local computation, and exploit it to design an efficient algorithm to solve the NDP.
We discuss the optimality of this procedure in the limit of infinitely large networks, and provide a sufficient condition for its optimality. We then provide numerical simulations, showing that
our algorithm achieves good performance even if the equilibrium support varies and the delay functions are non-linear.
\end{abstract}

\section{Introduction}
Due to increasing populations living in urban areas, many cities are facing the problem of traffic congestion, which leads to increasing levels of pollution and massive waste of time and money \cite{congestion}. 
The problem of mitigating congestion has been tackled in the literature from two main perspectives. 
One approach is to influence the user behaviour by incentive-design mechanisms, for instance by road tolling \cite{brown2017studies,fleischer2004tolls,zhao2006line,cole2006much,Como.ea:16,Como.Maggistro:22}, information design \cite{das2017reducing,meigs2020optimal,wu2019information,wu2021value} or lottery rewards \cite{yue2015reducing}, to minimize the inefficiencies due to the autonomous uncoordinated decisions of the agents. A second approach is to intervene on the transportation network, by building new roads or enlarging the existing ones. The corresponding \emph{network design problem} (i.e., the problem of optimizing the intervention on a transportation network subject to some budget constraints, see e.g. \cite{leblanc1975algorithm}) is very challenging because of its bi-level nature,
i.e., it involves a network intervention optimization problem given the flow distribution for that particular network. We assume that each link of the network is endowed with a delay function and the flow distributes according to a Wardrop equilibrium, taking paths with minimum cost, defined as the sum of the delay functions of the links along the path (see \cite{beckmann1956studies,wardrop1952road}). A characterization of the Wardrop equilibrium is used to construct the lower level of the bi-level network design problem.

In this work we define a network design problem (NDP), and analyze in details a special instance of the problem, where the delay functions are affine, and the planner can improve the delay function of a single link. Our objective is to strike a balance between a problem that is simple enough to guarantee tractable analysis, yet rich enough to allow insights for more general classes of NDPs. We then extend the validity of the proposed method by a numerical analysis, showing that good performance are achieved even if the delay functions are non-linear.
For single-link affine NDPs, our first theoretical result provides an analytical characterization of the cost variation (i.e., the total travel time at the equilibrium) corresponding to an intervention on a particular link under a regularity assumption, which states that the set of links carrying positive flow remain unchanged with an intervention. This assumption, which is not new in the traffic equilibrium literature (see e.g. \cite{steinberg1983prevalence,dafermos1984some}) 
leads to a characterization of Wardrop equilibria using a system of linear equations and enables representing single-link interventions as rank-1 perturbations of the system. We show that this assumption is satisfied provided that the total incoming flow to the network is large enough and the network is series-parallel, which may be of independent interest. We exploit the structure of our characterization and the linearity of the delay functions to express the cost variation using the effective resistance of a link (i.e., between the endpoints of the link), defined with respect to a related resistor network, obtained by making the directed transportation network undirected, and assigning a conductance to each link based on the delay function of the link. Computing the effective resistance of a single link requires the solution of a linear system whose dimension scales with the network size (we indistinctly refer to the network size as the cardinality of the node and the link sets, implicitly assuming that transportation networks are sparse in a such a way that the average degree of the nodes is independent of the number of nodes, inducing then a proportionality between the number of nodes and links). Hence, solving the NDP requires the solution of $\mathrm{E}$ of these problems, with $\mathrm{E}$ denoting the number of links. Since this can be computationally intractable for large networks, our second main result proposes a method based on Rayleigh's monotonicity laws to approximate the effective resistance of each link with a number of iterations independent of the network size, thus leading to a significant reduction of complexity.
The key idea is that the effective resistance between two adjacent nodes $i$ and $j$ depends mainly on the local structure of the network around the two nodes (i.e., the set of nodes $\mathcal{N}_{\le d}$ that are at distance no greater than a small constant $d$ from at least one of $i$ and $j$), and may therefore be approximate by performing only local computation. 
Since typically in transportation networks the local structure of the network is independent of the network size (think for instance of a bidimensional square grid), the size of $\mathcal{N}_{\le d}$ does not scale with the network size, thus we can guarantee that the approximation error and computational complexity of our method also do not scale.
Our third main result establishes sufficient conditions under which the approximation error vanishes asymptotically in the limit of infinite networks, proving that if the related resistor network is recurrent the approximation error tends to vanish for large distance $d$. In the conclusive section we conduct a numerical analysis on synthetic and real transportation networks, showing that a good approximation of the effective resistance of a link can be achieved by looking at a small portion of the network. Moreover, while several assumptions are made to establish theoretical results (e.g., affine delay functions, support of equilibrium flows not varying with the intervention), we conduct a numerical analysis showing that good performance are achieved even if some assumptions are relaxed, i.e., if the delay functions are non-linear and the support of the equilibrium is allowed to vary with interventions.

In our work we consider a special case of NDPs. These problems have been formalized in the last decades via many different formulations. Both \textit{continuous} network design problems \cite{chiou2005bilevel,li2012global,wang2014models},
where the budget can be allocated continuously among the links, and \textit{discrete} formulations, in which the decision variables include which new roads to build \cite{gao2005solution}, how many lanes to add to existing roads \cite{wang2013global}, or a mix of those two problems \cite{poorzahedy2007hybrid}, 
have been considered in the literature, together with \textit{dynamical} formulations \cite{fontaine2017dynamic}, and formulations where the optimum is achieved by removing, instead of adding, links, because of Braess' paradox \cite{roughgarden2006severity,fotakis2012efficient}. 
For comprehensive surveys on the literature on NDP we refer to \cite{yang1998models,farahani2013review}. We stress that most of the literature focuses on finding polynomial algorithms to solve in approximation NDPs in their most general form. Our main contribution is to provide a tractable approach to solve a single-link network design problem in quasi-linear time, as well as providing intuition and a completely new formulation. For the future we aim at extending our techniques to more general cases, like the multiple interventions case.
In the setting of affine delay functions, our NDP formulation is also related to the literature on marginal cost pricing. We assume that interventions modify the linear coefficient of the delay function of link $e$ from $a_e$ to $\tilde{a}_e$, leading to $\tilde{\tau}_e(f_e) = \tau_e(f_e)-(a_e-\tilde{a}_e)f_e$, which is equivalent to adding a negative marginal cost toll on a link. In the literature the problem of optimal toll design has been widely explored, also dealing with the problem of the support of the Wardrop equilibrium varying after the intervention, i.e., without imposing restrictive assumptions. However, most of the toll literature aim at finding conditions under which a general NP-hard problem may be solved in polynomial time. The scope of our work is instead to provide a new formulation to a more tractable problem. Moreover, to relax the regularity assumption on the support of the equilibrium, in the toll literature it is often assumed that the network has parallel links, which is unrealistic for transportation networks (see, e.g., \cite{hoefer2008taxing,jelinek2014computing}. 
Our work is also related to \cite{steinberg1983prevalence,dafermos1984some}, where the authors investigate the sign of total travel time variation when a new path is added to a two-terminal network, under similar assumptions to ours, providing sufficient conditions under which the Braess' paradox arises. In our work we instead compute the total travel time variation with an intervention, and suggest an efficient algorithm to select the optimal intervention.
As mentioned, the key step of our approach is to reformulate the NDP in terms of a resistance problem, and also exploits the parallelism between resistor networks and random walks. From a methodological perspective it is worthwhile mentioning that the relation between Wardrop equilibria and resistor networks has been first investigated in \cite{klimm2019computing}, while the parallelism between random walks and Wardrop equilibria has been investigated in \cite{rebeschini2018locality}, although with different purposes. The relation between random walks and resistor networks is quite standard and well-known (see e.g. \cite{doyle1984random}). To summarize, the contribution of this paper is two-fold. From a methodological perspective, we provide a method to locally approximate the effective resistance between adjacent nodes, which may be of independent interest (effective resistance of a link is related to spanning tree centrality \cite{hayashi2016efficient}). From NDP perspective, we provide a new formulation of the NDP in terms of resistor networks, and exploit our methodological result to approximate efficiently single-link NDPs.

The rest of the paper is organized as follows. In Section~\ref{sec:model} we define the model and formulate the NDP as a bi-level program. In Section~\ref{sec3} we define single-link NDPs, rephrase the problem in terms of resistor networks, and discuss the regularity assumption. In Section~\ref{approximation} we provide our method to approximate the effective resistance of a link and exploit such a method to construct an algorithm to solve the problem. We then analyze the asymptotic performance of the proposed method in the limit of infinite networks in Section~\ref{performance}. In Section~\ref{sec:simulations} we provide numerical simulations. Finally, in the conclusive section, we summarize the work and discuss future research lines.

\subsection{Notation}
We let $\delta^{(i)}$, $\mathbf{1}$, $\mathbf{0}$ and $\mb{I}$ denote the unitary vector with $1$ in position $i$ and $0$ in all the other positions, the column vector of all ones, the column vector of all zeros, and the identity matrix, respectively, where the size of them may be deduced from the context. $A^T$ and $v^T$ denote the transpose of matrix $A$ and vector $v$, respectively. Given a vector $v$, we let $\mb{I}_v$ denote the matrix whose off-diagonal elements are zero and with diagonal elements $(\mb{I}_v)_{ii} = v_i$. 

\section{Model and problem formulation}
\label{sec:model}
We model the transportation network as a directed multigraph $\mathcal{G}=(\mathcal{N},\mathcal{E})$, and denote by $\oo, \dd \in \mc{N}$ the origin and the destination of the network. We assume for simplicity of notation that $\mathcal{N}=\{1,\cdots,\mathrm{N}\}$ and $\mathcal{E}=\{1,\cdots,\mathrm{E}\}$, and assume that $\oo$ and $\dd$ are respectively the first and the last node of the network. Every link $e$ is endowed with a tail $\xi(e)$ and a head $\theta(e)$ in $\mc{N}$.  We allow multiple links between the same pair of nodes, and assume that every link belongs to at least a path from $\mathrm{o}$ to $\mathrm{d}$, otherwise such a link may be removed without loss of generality. 
Let $m > 0$ denote the throughput from the origin $\mathrm{o}$ to the destination $\mathrm{d}$, and $\nu =  m(\delta^{(\mathrm{o})}-\delta^{(\mathrm{d})}) \in \mathds{R}^\mathrm{N}$. Let $\mathcal{P}=\{1,\cdots,\mathrm{P}\}$ denote the set of paths from $\mathrm{o}$ to $\mathrm{d}$. 
An admissible path flow is a vector $z \in \mathds{R}^{\mathrm{P}}_+$ satisfying the mass constraint
\begin{equation}
	\label{constraints}
\mathbf{1}^T z=m.    
\end{equation}
Let $A \in \mathds{R}^{\mathrm{E} \times \mathrm{P}}$ denote the link-path incidence matrix, with entries $A_{ep}=1$ if link $e$ belongs to the path $p$ or $0$ otherwise. The path flow induces a unique link flow $f \in \mathds{R}^\mathrm{E}$ via
\begin{equation}
	f=Az.
	\label{incidence_path}
\end{equation}
Every link $e$ is endowed with a non-negative and strictly increasing delay function $\tau_e : \mathds{R}_+ \to \mathds{R}_+$. We assume that the delay functions are in the form $\tau_e(f_e) = \tau_e(0) + a_e(f_e)$, where $\tau_e(0)$ is the travel time of the link when there is no flow on it, and $a_e(f_e)$ describes congestion effects, with $a_e(0)=0$.
The cost of path $p$ under flow distribution $f$ is the sum of the delay functions of the links belonging to $p$, i.e.,
\begin{equation}
	c_p(f)=\sum_{e \in \mathcal{E}}A_{ep}\tau_e(f_e).
	\label{cost_path}
\end{equation}
\begin{definition}[Routing game]
	A \emph{routing game} is a triple $(\mathcal{G}, \tau, \nu)$.
\end{definition}

A Wardrop equilibrium is a flow distribution such that no one has incentive in changing path. More precisely, we have the following definition.
\begin{definition}[Wardrop equilibrium]
	A path flow $z^*$, with associated link flow $f^*=Az^*$, is a Wardrop equilibrium if for every path $p$
	\begin{equation*}
		z^*_p>0 \implies c_p(f^*)\le c_q(f^*), \quad \forall q \in \mathcal{P}.
	\end{equation*}
\end{definition}\medskip

Let $B \in \mathds{R}^{\mathrm{N} \times \mathrm{E}}$ denote the node-link incidence matrix, with entries $B_{ne}=1$ if $n=\xi(e)$, $B_{ne}=-1$ if $n=\theta(e)$, or $B_{ne}=0$ otherwise.
It is proved in \cite{beckmann1956studies} that a link flow $f^*$ is a Wardrop equilibrium of a routing game if and only if
\begin{equation}
	\begin{aligned}
		f^* = \ & \underset{f \in \mathds{R}_+^\mathrm{E}, Bf=\nu}{\arg\min}
		& & \sum_{e \in \mathcal{E}} \int_0^{f_e} \tau_e(s) ds,
	\end{aligned}
	\label{convex_prob}
\end{equation}
where $Bf=\nu$ is the projection of \eqref{constraints} on the link set.
Since the delay functions are assumed strictly increasing, the objective function in \eqref{convex_prob} is strictly convex and the Wardrop equilibrium $f^*$ is unique. 
\begin{definition}[Social cost]
	The social cost of a routing game is the total travel time at the equilibrium, i.e.,
	\begin{equation*}
		C^{(0)}=\sum_{e \in \mathcal{E}} f_e^*\tau_e(f_e^*).
	\end{equation*}
\end{definition}
\medskip
The social cost can be interpreted as a measure of performance by a planner that aims at minimizing the overall congestion on the transportation network.
We now provide an equivalent characterization of the social cost of a routing game. To this end, let $\lambda^*$ and $\gamma^*$ denote the Lagrangian multipliers associated to $f^*\ge \mathbf{0}$ and $Bf=\nu$, respectively.
The KKT conditions of \eqref{convex_prob} read:
\begin{align}
	\label{kkt2}
	\begin{cases}
		\tau_e(f_e^*)+\gamma_{\theta(e)}^*-\gamma^*_{\xi(e)} - \lambda^*_{e}=0 & \forall e \in \mathcal{E},\\
		\sum_{e \in \mathcal{E}: \theta(e)=i} f_e - \sum_{e \in \mathcal{E}: \xi(e)=i} f_e + \nu_i=0 & \forall i \in \mathcal{N},\\
		\lambda_{e}^*f_{e}^*=0 & \forall e \in \mathcal{E},\\
		\lambda_{e}^*\ge 0 & \forall e  \in \mathcal{E},\\
		f_{e}^*\ge 0 & \forall e \in \mathcal{E}.
	\end{cases}
\end{align}
The third condition, known as complementary slackness, implies that if $\lambda_e^*>0$, then $f_e^*=0$, i.e., link $e$ is not used at the equilibrium. We let $\mathcal{E}_+$ denote the set of the links $e$ such that $\lambda_e^*>0$. The next lemma shows that the social cost may be characterized in terms of the Lagrangian multiplier $\gamma^*$.
\begin{lemma}
	\label{cost}
	Let $(\mc{G},\tau,\nu)$ denote a routing game. Then, 
	\begin{equation*}
		\begin{gathered}
			C^{(0)}=m(\gamma^*_\oo-\gamma_\dd^*).
		\end{gathered}
	\end{equation*}
\end{lemma}\medskip
\begin{proof}
	See Appendix~\ref{app:proofs}.
\end{proof}

We consider a NDP where the planner can improve the delay functions of the network with the goal of minimizing a combination of the social cost after the intervention and the cost of the intervention itself. Specifically, let $\control \in \mathds{R}_+^{\mathrm{E}}$ denote the intervention vector, with corresponding delay functions
$$
\tau_e^{(\control_e)}(f_e) = \tau_e(0) + \frac{a_e(f_e)}{1+\control_e}.
$$
This type of interventions may correspond for instance to enlarging some roads of the network. We let $h_e : [0,+\infty) \to [0,+\infty)$ denote the cost associated to the intervention on link $e$. The goal of the planner is to minimize a combination of the social cost and the intervention cost, where $\alpha \ge 0$ is the trade-off parameter. More precisely, by letting $f^*(\control)$ denote the Wardrop equilibrium corresponding to intervention $\control$, the NDP reads as follows.
\begin{problem}
	\label{prob:general}
	Let $(\mc{G},\tau,\nu)$ be a routing game, and $\alpha\ge0$ be the trade-off parameter. The goal is to select $\control^*$ such that
\begin{equation}
	\label{eq:ndp}
		\control^* \in \underset{\control \in \mathds{R}^\mathrm{E}_+}{\arg\min} \
		\sum_{e \in \mc{E}} f^*_e(\control) \tau_e^{(\control_e)}(f^*_e(\control)) +  \alpha h(\control),
\end{equation}
where $h(\control)=\sum_e h_e(\control_e)$, and
\begin{equation}
	\begin{aligned}
	f^*(\control) = \ & \underset{f \in \mathds{R}_+^\mathrm{E}, Bf=\nu}{\arg\min}
	& & \sum_{e \in \mathcal{E}} \int_0^{f_e} \tau_e^{(\control_e)}(s) ds.
	\end{aligned}
\label{eq:f(k)}
\end{equation}
\end{problem}
\medskip
\begin{remark}
We stress the fact that Problem~\ref{prob:general} is bi-level, in the sense that the planner optimizes the intervention $\control$ according to a cost function that depends on the Wardrop equilibrium $f^*(\control)$, which in turn is the solution of the optimization problem \eqref{eq:f(k)}, whose objective function depends on the intervention $\control$ itself.
\end{remark}

\begin{remark}
Problem~\ref{prob:general} is not equivalent to the toll design problem. The key difference between the two problems is that tolls modify the Wardrop equilibrium, but the performance of tolls is evaluated with respect to the original delay functions $\tau_e$. On the contrary, in Problem~\ref{prob:general} the intervention is evaluated with respect to the new delay functions $\tau_e^{(\control_e)}$.
\end{remark}

Problem~\ref{prob:general} is in general non-convex, and hard to solve because of its bi-level nature. For these reasons, in the next section we shall study a simplified problem where the delay functions are affine and the planner may intervene on one link only. In this setting we are able to rephrase the problem as a single-level optimization problem, and provide an electrical network interpretation of the problem.

\section{Single-link interventions in affine networks}
\label{sec3}
In this section we provide an electrical network formulation of the NDP under some restrictive assumptions. In particular, we provide a closed formula for the social cost variation in terms of electrical quantities computed on a related resistor network. To this end, we restrict our analysis to the space of feasible interventions $\Control$, defined as
$$
\Control := \{\control :  \control_e \delta^{(e)} \ \text{for a link} \ e \in \mc{E}, \control_e \ge 0  \}.
$$
In other words, $\Control$ represents the space of interventions on a single link of the network. We also assume that the delay functions are affine, i.e., $\tau_e(f_e) = a_ef_e + b_e$ for every $e$, and denote by $(\mc{G},a,b,\nu)$ routing games with affine delay functions.
For an intervention $\control$, let $(\mc{G},a^{(\control)},b,\nu)$ denote the corresponding affine routing game, $C^{(\control)}$ denote the corresponding social cost, and $\Delta C^{(\control)} = C^{(0)} - C^{(\control)}$ denote the social cost gain. Our problem can be expressed as follows.
\begin{problem}
	Let $(\mathcal{G}, a, b, \nu)$ be an affine routing game and $\alpha \ge0$ be the trade-off parameter. Find
	\begin{equation*}
			\control^* \in \ \underset{\control \in \Control}{\arg\max}
			 \ (\Delta C^{(\control)} - \alpha h(\control)).
	\end{equation*}
	\label{prob1}
\end{problem}

The next example shows that the problem cannot be decoupled by first selecting the optimal link $e^*$ and then the optimal strength of the intervention $\control^*_e$. 
\begin{example}
	\label{ex:ex}
	Consider the transportation network in Figure~\ref{fig:ex}, with linear delay functions
	$\tau_e(f_e) = a_ef_e$.
	By some computation, one can prove that
	\begin{equation*}
		\begin{aligned}
			\Delta C^{(\control_1\delta^{(1)})} & = m \frac{a_1a_2^2 \control_1}{(a_1+a_2)((\control_1+1) a_1 + a_2)},\\
			\Delta C^{(\control_2 \delta^{(2)})} & = m \frac{a_1^2a_2 \control_2}{(a_1+a_2)(a_1 + (\control_2+1) a_2)},\\
			\Delta C^{(\control_3 \delta^{(3)})} & = m \frac{\control_3}{\control_3+1}.
		\end{aligned}
	\end{equation*}
	In Figure~\ref{fig:ex} the social cost variation corresponding to intervention on every link $e$ is illustrated as functions of $\control_e$. Observe that the link that maximizes the social cost gain depends on $\control_e$. Thus, 
	the problem cannot be decoupled by first selecting the optimal link $e^*$ and then the optimal $\control_e^*$.
	\begin{figure}
		\centering
		\begin{tikzpicture}
			\node[draw, circle] (1) at (0,0.5)  {o};
			\node[draw, circle] (2) at (0,-1)  {n};
			\node[draw, circle] (3) at (0,-2.5) {d};
			\node (4) at (0,-3) {};
			
			\path [->, >=latex]  (1) edge [bend right=30]
			node [left] {$e_1$} (2);
			\path [->, >=latex]  (1) edge [bend left=30]
			node [right] {$e_2$} (2);
			\path [->, >=latex]  (2) edge [bend left=0]
			node [right] {$e_3$} (3);
		\end{tikzpicture}
		\includegraphics[width = 6cm]{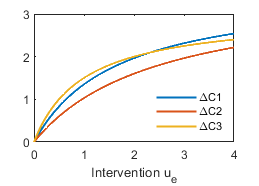}
		\caption{\emph{Left}: the graph of Example~\ref{ex:ex}. \emph{Right}: the social cost variation corresponding to single link interventions, with assignment $a_1 = 3, a_2 = 2, a_3 = 1, m = 3$.
			\label{fig:ex}}
	\end{figure}
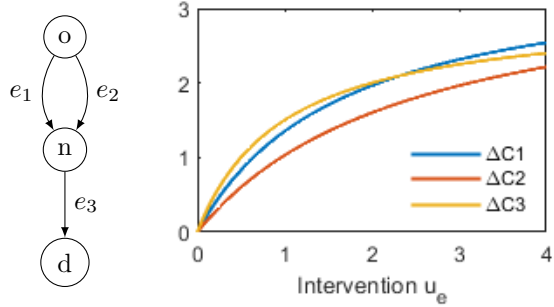
\end{example}

Our theoretical results rely on the following technical assumption, stating that the support of the Wardrop equilibrium is not modified with an intervention.
\begin{assumption}
	Let $\mathcal{E}_+(\control)$ be the set of links $e$ such that for the routing game $(\mathcal{G}, a^{(\control)}, b, \nu)$ the Lagrangian multiplier $\lambda_e^*(\control)>0$.
	We assume that $\mathcal{E}_+(\control)=\mathcal{E}_+$ for every $\control$ in $\Control$.
	\label{assumption}
\end{assumption}

Assumption~\ref{assumption} is not new in the literature \cite{steinberg1983prevalence,dafermos1984some}.
We will get back to the assumption in Section~\ref{sec:assumption}.
With a slight abuse of notation, from now on let $\mc{E}$ denote $\mc{E} \setminus \mc{E}_+$.
We now define a mapping from the transportation network $\mc{G}$ to an associated resistor network $\mc{G}_R$.
\begin{definition}[Associated resistor network]
	Given the transportation network $\mc{G}=(\mc{N},\mc{E})$, the \emph{associated resistor network} $\mathcal{G}_R=(\mc{N},\mc{L},W)$ is constructed as follows:
		\begin{itemize}
			\item the node set $\mc{N}$ is the same.
			\item $W \in \mathds{R}^{\mathrm{N} \times \mathrm{N}}$ is the conductance matrix, with elements
			\begin{equation} 
				W_{ij}= \begin{cases}
					\sum_{\substack{e \in \mathcal{E}: \\
							\xi(e)=i, \theta(e)=j, \ \text{or} \\
							\xi(e)=j, \theta(e)=i}} \frac{1}{a_e} & \text{if}\ i \neq j \\
					0 & \text{if}\ i=j. 
				\end{cases}
				\label{W}
			\end{equation}
			Note that $W$ is symmetric, thus $\mc{G}_R$ is undirected. The element $W_{ij}$ has to be interpreted as the conductance between nodes $i$ and $j$.
			\item Multiple links connecting the same pair of nodes are not allowed, hence every link $l$ in $\mc{L}$ can be identified by a unordered pair of nodes $\{i,j\}$, and the set $\mc{L}$ is uniquely determined by $W$. Let $\mathrm{L}$ denote the cardinality of $\mc{L}$. The mapping $M: \mc{E} \to \mc{L}$ associates to every link $e$ of the transportation network the corresponding link $l = M(e) = \{\xi(e),\theta(e)\}$ of the resistor network. Note by \eqref{W} that $M(e)$ belongs to $\mc{L}$ for every $e$ in $\mc{E}$. 
		\end{itemize}
\end{definition}

Note that the coefficients $a_e$ correspond to resistances in the resistor networks. We let $w = W\mb{1}$ denote the degree distribution of the resistor network, and $w^* = \max_{i \in \mc{N}}w_i$ denote the maximal degree. Before establishing our first main result, we define two relevant quantities.
\begin{definition}
	\label{def:voltage}
	Let $v \in \mathds{R}^\mathrm{N}$ be the \emph{voltage} vector on $\mc G_R$ when a net electrical current $m$ is injected from $\oo$ to $\dd$, i.e., $v$ is the unique solution of
\begin{equation}
	\label{eq:voltage}
	\sum_{k \in \mc{N}} W_{hk}(v_h- v_k)=m(\delta^{(\oo)}-\delta^{(\dd)})\qquad \forall h\in\mathcal N.
\end{equation}
For a link $e$ in $\mc{E}$, let $y_e$ denote the \emph{electrical current} flowing from $\xi(e)$ to $\theta(e)$ on link $M(e)$ of $\mc G_R$, and let $\Delta v_e = v_{\xi(e)}-v_{\theta(e)}$. By Ohm's law, $\Delta v_e = a_e y_e$.
\end{definition}
\begin{definition}
	\label{def:effective}
	Let $\ov v \in \mathds{R}^\mathrm{N}$ be the voltage vector on $\mc G_R$ when a unitary current is injected from $i$ to $j$, i.e., 
	\begin{equation}
		\label{eq:potential}
		\sum_{k \in \mc{N}} W_{hk}(\ov v_h- \ov v_k)=\delta^{(i)}-\delta^{(j)}\qquad \forall h\in\mathcal N.
	\end{equation}
    The \emph{effective resistance} $r_{l}$ of link $l = \{i,j\}$ in $\mc L$ is the effective resistance between $i$ and $j$, i.e., $r_{l}=\ov{v}_i-\ov{v}_j$. Given a link $e$ in $\mc E$, we denote by $r_e$ the effective resistance of link $M(e)$ of the associated resistor network.
\end{definition}

The next theorem establishes a relation between the social cost gain with a single-link intervention and the associated resistor network. 
\begin{theorem}
	Let $(\mathcal{G}, a, b, \nu)$ be an affine routing game, and let Assumption~\ref{assumption} hold. Then,
	\begin{equation}
		 \Delta C^{(\control_e \delta^{(e)})}=a_e f_e^*\frac{y_e}{\frac{1}{\control_e}+\frac{r_e}{a_e}}.
		 		\label{el_form}
	\end{equation}
	\label{thm}
\end{theorem}
\begin{proof}
See Appendix~\ref{app:proofs}.
\end{proof}

The ratio $r_{e}/a_e$ belongs to $(0,1]$ and is also known as \emph{spanning tree centrality}, which measures the fraction of spanning trees including link $M(e)$ among all spanning trees of the undirected network $\mc G_R$ \cite{hayashi2016efficient}. The spanning tree centrality of a link is maximized when removing the link disconnects the network. 
Theorem \ref{thm} states that the social cost variation due to intervention on link $e$ is: 
	\begin{itemize}
		\item proportional to $a_e f_e^*$, which measures the delay at the equilibrium due to congestion on link $e$;
		\item decreasing in the spanning tree centrality. Intuitively speaking, the benefits of intervention on link $e$ is larger when the intervention modifies the equilibrium flows so that agents can move from paths not including $e$ to paths including $e$, namely when $f_e^*$ increases after the intervention. This phenomenon does not occur if $e$ is a bridge, i.e., if $r_{e}/a_e=1$, and occurs largely when many paths from $\xi(e)$ to $\theta(e)$ exist, i.e., when $r_{e}/a_e$ is small;
		\item proportional to the current $y_e$. The role of this term is more clear in the special case of linear delay functions. In this case $y_e = f_e^*$ for all links $e$ in $\mc{E} \setminus \mc E_+ $, hence $a_e f_e^* y_e^* = a_e (f_e^*)^2$, which is the total travel time on link $e$ before the intervention. 
	\end{itemize}
The idea behind the proof is that with affine delay functions the KKT conditions of the Wardrop equilibrium are linear, and under Assumption~\ref{assumption} single-link interventions are equivalent to rank-1 perturbations of the system. Thus, by Lemma~\ref{cost} we can compute the cost variation by looking at Lagrangian multiplier $\gamma_\oo^*$, and then express such a variation in terms of electrical quantities.
In order to solve Problem~\ref{prob1} by the electrical formulation, we need to compute \eqref{el_form} for every link $e$ in $\mc E$. The Wardrop equilibrium $f^*$ is assumed to be observable and therefore given. The voltage $v$ (and thus $y$) can be derived by solving the linear system \eqref{eq:voltage} and has to be computed only once.
On the contrary, the computation of $r_{e}$ must be repeated for every link, hence it requires to solve $\mathrm{L}$ sparse linear systems. To reduce the computational effort, in Section~\ref{approximation} we shall propose a method to \textit{approximate} the effective resistance of a link that, under a suitable assumption on the sparseness of the network, does not scale with the network size, allowing for a more efficient solution to Problem~\ref{prob1}. The next result shows how to compute the derivative of the social cost variation for small interventions.
\begin{corollary}
Let $(\mathcal{G}, a, b, \nu)$ be a routing game, and assume that for every $i$ in $\mc E$ it holds either $f_i^*>0$ or $\lambda_i^*>0$. Then,
$$
\frac{\partial \Delta C(\control)}{\partial \control_e}\Big|_{\control = \mb{0}} = 
\begin{cases}
a_e f_e^* y_e \quad & \text{if} \ \lambda_e^* = 0, \\
0 & \text{if} \ \lambda_e^* > 0.
\end{cases}
$$
\end{corollary}
\medskip
\begin{proof}
The fact that for every link $i$ it holds either $f_i^*>0$ or $\lambda_i^*>0$ implies that for infinitesimal interventions the support of $f^*$ is not modified. If $\lambda_e^*=0$, then $f_e^*$ and we can derive the social cost variation in \eqref{el_form} with respect to $\control_e$. The case $\lambda_e^*>0$ follows from continuity arguments and from the complementary slackness condition, which implies that $f_e^*(\control)=0$ in a neighborhood of $\control = \mb{0}$.
\end{proof}
\begin{remark}
Observe that the derivative of the social cost does not depend on the effective resistance of the link.
\end{remark}

\subsection{On the validity of Assumption~\ref{assumption}}
\label{sec:assumption}
In this section we discuss Assumption~\ref{assumption}. In particular, we show that the assumption is without loss of generality on series-parallel networks, if the throughput is sufficiently large. We first recall the definition of directed series-parallel networks, and then present the result in Proposition~\ref{sp_prp}.
\begin{definition}
	A directed network $\mathcal{G}$ is series-parallel if and only if
	(i) it is composed of two nodes only ($\oo$ and $\dd$), connected by single link from $\oo$ to $\dd$, or
	(ii) it is the result of connecting two directed series-parallel networks $\mathcal{G}_1$ and $\mathcal{G}_2$ in parallel, by merging $\mathrm{o}_1$ with $\mathrm{o}_2$ and  $\mathrm{d}_1$ with $\mathrm{d}_2$, or
	(iii) it is the result of connecting two directed series-parallel networks $\mathcal{G}_1$ and $\mathcal{G}_2$ in series, by merging $\mathrm{d}_1$ with $\mathrm{o}_2$.
\end{definition}

\begin{proposition}
	\label{sp_prp}
	Let $(\mathcal{G}, a, b, \nu)$ be a routing game. If $\mathcal{G}$ is series-parallel, there exists $\overline{m}$ such that for every $m \ge \overline{m}$, $\mathcal{E}_+=\emptyset$. Furthermore, if $b=\mathbf{0}$, $\mathcal{E}_+=\emptyset$ for every $m >0$.
\end{proposition}
\begin{proof}
	See Appendix~\ref{app:proofs}.
\end{proof}

\begin{remark}
	Proposition~\ref{sp_prp} implies that Assumption~\ref{assumption} is without loss of generality on series-parallel networks if $m \ge \overline{m}$. 
\end{remark}

The next example shows that, if the throughput is not sufficiently large, Assumption~\ref{assumption} may be violated. 
\begin{example}
	Consider the series-parallel network in Figure~\ref{wheatstone}.
	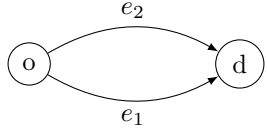
\begin{figure}
		\centering
		\begin{tikzpicture}[scale = 0.7]
			\node[draw, circle] (1) at (0,0)  {o};
			\node[draw, circle] (2) at (4,0) {d};
			\path [->, >=latex]  (1) edge [bend right=30]
			node [below] {$e_1$} (2);
			\path [->, >=latex]  (1) edge [bend left=30]
			node [above] {$e_2$} (2);
		\end{tikzpicture}
		\caption{A directed series-parallel network. If the throughput is not sufficiently large, Assumption~\ref{assumption} is not guaranteed to hold.}
		\label{wheatstone}
	\end{figure}
	Let $m=1$, and consider affine delay functions $\tau_1(f_1)=f_1+1, \tau_2(f_2)=f_2+3/2$. One can verify that
	\begin{equation*}
		f_1^*=3/4, \quad f_2^*=1/4, \quad \lambda_1^*=\lambda_2^*=0.
	\end{equation*}
	Modifying $a_1$ from $1$ to $1/3$ (i.e., with $\control=2\delta^{(1)}$), we get:
	\begin{equation*}
		f_1^*(\control)=1, \quad f_2^*(\control)=0, \quad \lambda_1^*(\control)=0, \quad \lambda_2^*(\control)=1/6,
	\end{equation*}
	violating Assumption~\ref{assumption}. Proposition~\ref{sp_prp} proves that this does not occur if $m$ is sufficiently large.
\end{example}

\section{An approximate solution to Problem 1}
\label{approximation}
As shown in the previous section, Problem~\ref{prob1} may be rephrased in terms of electrical quantities over a related resistor network. Solving the NDP problem in this formulation requires to solve $\mathrm{L}$ linear systems whose dimension scales linearly with $\mathrm{N}$. Since the voltage $v$ may be computed in quasi-linear time by solving the sparse linear system \eqref{eq:voltage} (see \cite{cohen2014solving} for more details), the computational bottleneck is given by the computation of the effective resistance of every link of the resistor network. 
The main idea of our method is that, although the effective resistance of a link depends on the entire network, it can be approximate by looking at a local portion of the network only.
We then formulate an algorithm to solve Problem~\ref{prob1} by exploiting our approximation method.
\subsection{Approximating the effective resistance}
\label{method}
We introduce the following operations on resistor networks.
\begin{definition}[Cutting at distance $d$]
	\label{def:cutting}
	A resistor network $\mathcal{G}_R$ is cut at distance $d$ from link $l = \{i,j\}$ in $\mc{L}$ if every node at distance greater than $d$ from link $l$ (i.e., from both $i$ and $j$) is removed, and every link with at least one endpoint in the set of the removed nodes is removed. Let $\mathcal{G}^{U_{d}}_{l}$ and $r_{l}^{U_d}$ denote such a network and the effective resistance of link $l$ on it, respectively.
\end{definition}

\begin{definition}[Shorting at distance $d$]
	A resistor network $\mathcal{G}_R$ is shorted at distance $d$ from $l$ in $\mc{L}$ if all the nodes at distance greater than $d$ from link $l$ are shorted together, i.e., an infinite conductance is added between each pair of such nodes. Let $\mathcal{G}^{L_{d}}_{l}$ and $r_{l}^{L_d}$ denote such a network and the effective resistance of link $l$ on it, respectively.
\end{definition}

We refer to Figure~\ref{quad} for an example of these techniques applied to a regular grid.
We next prove that $r^{U_d}_{l}$ and $r^{L_d}_{l}$ are respectively an upper and a lower bound for the effective resistance $r_{l}$ for every link $l$. To this end, let us introduce Rayleigh's monotonicity laws.
\begin{figure}
	\centering
	\begin{tikzpicture}[scale=0.7, transform shape]
		\node[draw,fill=green] (1) at (0,0)  {};
		\node[draw,fill=green] (2) at (1,0)  {};
		\node[draw,fill=yellow] (3) at (-1,0) {};
		\node[draw,fill=yellow] (4) at (0,1) {};
		\node[draw,fill=yellow] (5) at (1,1)  {};
		\node[draw,fill=yellow] (6) at (2,0)  {};
		\node[draw,fill=yellow] (7) at (1,-1)  {};
		\node[draw,fill=yellow] (8) at (0,-1)  {};
		\node[draw,fill=orange] (9) at (-2,0) {};
		\node[draw,fill=orange] (10) at (-1,1) {};
		\node[draw,fill=orange] (11) at (0,2)  {};
		\node[draw,fill=orange] (12) at (1,2)  {};
		\node[draw,fill=orange] (13) at (2,1)  {};
		\node[draw,fill=orange] (14) at (3,0)  {};
		\node[draw,fill=orange] (15) at (2,-1)  {};
		\node[draw,fill=orange] (16) at (1,-2)  {};
		\node[draw,fill=orange] (17) at (0,-2)  {};
		\node[draw,fill=orange] (18) at (-1,-1)  {};
		\node[draw,fill=red] (19) at (-3,0) {};
		\node[draw,fill=red] (20) at (-2,1) {};
		\node[draw,fill=red] (21) at (-1,2)  {};
		\node[draw,fill=red] (22) at (0,3)  {};
		\node[draw,fill=red] (23) at (1,3)  {};
		\node[draw,fill=red] (24) at (2,2)  {};
		\node[draw,fill=red] (25) at (3,1)  {};
		\node[draw,fill=red] (26) at (4,0)  {};
		\node[draw,fill=red] (27) at (3,-1)  {};
		\node[draw,fill=red] (28) at (2,-2)  {};
		\node[draw,fill=red] (29) at (1,-3) {};
		\node[draw,fill=red] (30) at (0,-3) {};
		\node[draw,fill=red] (31) at (-1,-2)  {};
		\node[draw,fill=red] (32) at (-2,-1)  {};
		\node[draw,fill=none] (33) at (-3,1)  {};
		\node[draw,fill=none] (34) at (-2,2)  {};
		\node[draw,fill=none] (35) at (-1,3)  {};
		\node[draw,fill=none] (36) at (2,3)  {};
		\node[draw,fill=none] (37) at (3,2)  {};
		\node[draw,fill=none] (38) at (4,1)  {};
		\node[draw,fill=none] (39) at (4,-1) {};
		\node[draw,fill=none] (40) at (3,-2) {};
		\node[draw,fill=none] (41) at (2,-3)  {};
		\node[draw,fill=none] (42) at (-1,-3)  {};
		\node[draw,fill=none] (43) at (-2,-2)  {};
		\node[draw,fill=none] (44) at (-3,-1)  {};
		\node[draw,fill=none] (45) at (-3,2)  {};
		\node[draw,fill=none] (46) at (-2,3)  {};
		\node[draw,fill=none] (47) at (3,3)  {};
		\node[draw,fill=none] (48) at (4,2)  {};
		\node[draw,fill=none] (49) at (4,-2) {};
		\node[draw,fill=none] (50) at (3,-3) {};
		\node[draw,fill=none] (51) at (-2,-3)  {};
		\node[draw,fill=none] (52) at (-3,-2)  {};
		\node[draw,fill=none] (53) at (-3,3)  {};
		\node[draw,fill=none] (54) at (4,3)  {};
		\node[draw,fill=none] (55) at (4,-3)  {};
		\node[draw,fill=none] (56) at (-3,-3)  {};
		\path (1) edge node [above] {$l$} (2);
		\path (1) edge (3);
		\path (1) edge (4);
		\path (1) edge (8);
		\path (2) edge (5);
		\path (2) edge (6);
		\path (2) edge (7);
		\path (4) edge (5);
		\path (7) edge (8);
		\path (3) edge (9);
		\path (3) edge (10);
		\path (10) edge (4);
		\path (4) edge (11);
		\path (5) edge (12);
		\path (11) edge (12);
		\path (3) edge (9);
		\path (3) edge (10);
		\path (10) edge (4);
		\path (5) edge (13);
		\path (6) edge (13);
		\path (6) edge (14);
		\path (6) edge (15);
		\path (15) edge (7);
		\path (16) edge (7);
		\path (16) edge (17);
		\path (8) edge (17);
		\path (8) edge (18);
		\path (3) edge (18);
		\path (19) edge (9);
		\path (9) edge (20);
		\path (10) edge (20);
		\path (10) edge (21);
		\path (11) edge (21);
		\path (11) edge (22);
		\path (22) edge (23);
		\path (12) edge (23);
		\path (12) edge (24);
		\path (24) edge (13);
		\path (25) edge (13);
		\path (25) edge (14);
		\path (26) edge (14);
		\path (27) edge (14);
		\path (15) edge (28);
		\path (15) edge (27);
		\path (16) edge (28);
		\path (16) edge (29);
		\path (29) edge (30);
		\path (17) edge (30);
		\path (17) edge (31);
		\path (31) edge (18);
		\path (32) edge (9);
		\path (32) edge (18);
		\path (53) edge (46);
		\path (46) edge (35);
		\path (35) edge (22);
		\path (23) edge (36);
		\path (36) edge (47);
		\path (47) edge (54);
		\path (45) edge (34);
		\path (34) edge (21);
		\path (24) edge (37);
		\path (37) edge (48);
		\path (33) edge (20);
		\path (25) edge (38);
		\path (44) edge (32);
		\path (27) edge (39);
		\path (52) edge (43);
		\path (43) edge (31);
		\path (28) edge (40);
		\path (40) edge (49);
		\path (56) edge (51);
		\path (51) edge (42);
		\path (42) edge (30);
		\path (29) edge (41);
		\path (41) edge (50);
		\path (50) edge (55);
		\path (53) edge (45);
		\path (46) edge (34);
		\path (35) edge (21);
		\path (24) edge (36);
		\path (37) edge (47);
		\path (48) edge (54);
		\path (45) edge (33);
		\path (34) edge (20);
		\path (25) edge (37);
		\path (38) edge (48);
		\path (33) edge (19);
		\path (26) edge (38);
		\path (44) edge (19);
		\path (26) edge (39);
		\path (52) edge (44);
		\path (43) edge (32);
		\path (27) edge (40);
		\path (39) edge (49);
		\path (56) edge (52);
		\path (51) edge (43);
		\path (42) edge (31);
		\path (28) edge (41);
		\path (40) edge (50);
		\path (49) edge (55);
	\end{tikzpicture}
	\vspace{0.3cm} \\
	\centering{$\mathcal{G}_{l}^{U_1}$ \hspace{2.6cm} $\mathcal{G}_{l}^{L_1}$} \\
	\vspace{0.2cm}
		\centering
		\begin{tikzpicture}[scale=0.85, transform shape]
			\node[draw,fill=green] (1) at (0,0)  {};
			\node[draw,fill=green] (2) at (1,0)  {};
			\node[draw,fill=yellow] (3) at (-1,0) {};
			\node[draw,fill=yellow] (4) at (0,1) {};
			\node[draw,fill=yellow] (5) at (1,1)  {};
			\node[draw,fill=yellow] (6) at (2,0)  {};
			\node[draw,fill=yellow] (7) at (1,-1)  {};
			\node[draw,fill=yellow] (8) at (0,-1)  {};
			\path (1) edge node [above] {$l$} (2);
			\path (1) edge (3);
			\path (1) edge (4);
			\path (1) edge (8);
			\path (2) edge (5);
			\path (2) edge (6);
			\path (2) edge (7);
			\path (4) edge (5);
			\path (7) edge (8);
		\end{tikzpicture}
	\quad
		\begin{tikzpicture}[scale=0.65, transform shape]
			\node[draw,fill=green] (1) at (0,0)  {};
			\node[draw,fill=green] (2) at (1,0)  {};
			\node[draw,fill=yellow] (3) at (-1,0) {};
			\node[draw,fill=yellow] (4) at (0,1) {};
			\node[draw,fill=yellow] (5) at (1,1)  {};
			\node[draw,fill=yellow] (6) at (2,0)  {};
			\node[draw,fill=yellow] (7) at (1,-1)  {};
			\node[draw,fill=yellow] (8) at (0,-1)  {};
			\node[draw,fill=orange] (9) at 
			(0.5,2) {s};
			\path (1) edge node [above] {$l$} (2);
			\path (1) edge (3);
			\path (1) edge (4);
			\path (1) edge (8);
			\path (2) edge (5);
			\path (2) edge (6);
			\path (2) edge (7);
			\path (4) edge (5);
			\path (7) edge (8);
			\path (4) edge [bend right=0]
			node [below] {} (9);
			\path (5) edge [bend right=0]
			node [below] {} (9);
			\path (3) edge [bend right=-20]
			node [below] {} (9);
			\path (6) edge [bend right=20]
			node [below] {} (9);
			\coordinate (ghost) at (3, 0);
			\coordinate (ghost2) at (-2,0);
			\draw[-] (7) to[out=0, in=-90] (ghost) to[out=90, in=0] (9);
			\draw[-] (8) to[out=180, in=-90] (ghost2)
			to[out=90, in=180] (9);
		\end{tikzpicture}
	\caption{Square grid. \emph{Above}: the yellow, orange and red nodes are at distance $1$, $2$ and $3$, respectively from the green nodes.
	\emph{Bottom-left}: the grid cut at distance $1$ from link $l$.
		\emph{Bottom-right}: the grid shorted at distance $1$ from link $l$. 
		Note that in the bottom right network the links connecting yellow nodes with node $s$ do not have unitary weights.}
	\label{quad}
\end{figure}
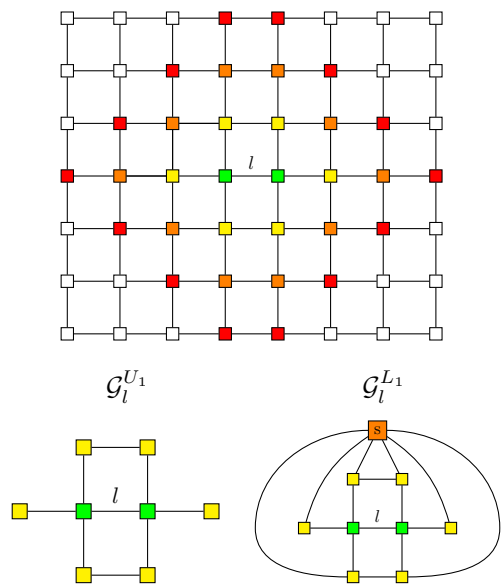
\begin{lemma}[Rayleigh's monotonicity laws \cite{levin2017markov}]
	If the resistances of one or more links are increased, the effective
	resistance between two arbitrary nodes cannot decrease.
	If the resistances of one or more links are decreased, the effective resistance cannot increase.
	\label{ray}
\end{lemma}
\begin{proposition}
		\label{prp_ray}
	Let $\mathcal{G}_R$ be a resistor network. For every link $l=\{i,j\}$ in $\mc{L}$, 
	\begin{equation*}
		\label{bounds}
		r^{U_{d_1}}_{l}\ge r^{U_{d_2}}_{l} \ge r_{l} \ge r^{L_{d_2}}_{l} \ge r^{L_{d_1}}_{l}, \quad \forall d_2 \ge d_1 \ge 1.
	\end{equation*}
	Moreover,
	\begin{equation}
		1/w^* \le r^{L_d}_{l}\le r^{U_d}_{l} \le 1/W_{ij}, \quad \forall d \ge 1.
		\label{remark}
	\end{equation}
\end{proposition}\medskip
\begin{proof}
Cutting a network at distance $d$ is equivalent to setting to infinity the resistance of all the links with at least one endpoint at distance greater than $d$. Shorting a network at distance $d$ is equivalent to setting to zero the resistance between any pair of nodes at distance greater than $d$. Then, by Rayleigh's monotonicity laws, it follows $r^{U_d}_{l}\ge r_{l} \ge r^{L_d}_{l}$. Similar arguments may be used to show that, if $d_1<d_2$, then $r^{U_{d_1}}_{l} \ge r^{U_{d_2}}_{l}$ and $r^{L_{d_1}}_{l} \le r^{L_{d_2}}_{l}$. The right inequality in \eqref{remark} follows from Rayleigh's monotonicity laws, by noticing that the effective resistance computed in the network with only nodes $i$ and $j$ (which is equal to $1/W_{ij}$) is an upper bound for $r_{l}^{U_{1}}$. The left inequality follows from noticing that the effective resistance on the network in which every node except $i$ is shorted with $j$, which results in a network with only two nodes and a conductance between $i$ and $j$ not greater than $w^*$ (hence, resistance no less than $1/w^*$) is a lower bound of $r^{L_1}_{l}$.
\end{proof}

Proposition~\ref{prp_ray} states that cutting and shorting a network provides upper and lower bound for the effective resistance of a link. Moreover, the bound gap is a monotone function of the distance $d$.

\subsection{Our algorithm}
We here propose an algorithm to solve in approximation Problem~\ref{prob1} based on our method for approximating the effective resistance. Our approach is detailed in Algorithm 1.
\begin{algorithm}
	\SetAlgoNoLine
	\KwIn{The affine routing game $(\mathcal{G},a,b,\nu)$, the cost functions $\{h_e\}_{e \in \mc{E}}$, and the distance $d\ge1$ for effective resistance approximation.}
	\KwOut{The optimal intervention $\control^{*d}$.}
	Construct the associated resistor network $\mc{G}_R$.\\
	Compute $v$ and $y$ by solving \eqref{eq:voltage}.\\
	\For{every $l$ in $\mathcal{L}$}
	{Construct $\mathcal{G}_{l}^{U_d}$ and $\mathcal{G}_{l}^{L_d}$. \\
		Compute $r_{l}^{U_d}$ and $r_{l}^{L_d}$ on $\mathcal{G}_{l}^{U_d}$ and $\mathcal{G}_{l}^{L_d}$.}
	\For{every $e$ in $\mathcal{E}$}
	{Find $\control_e^{*d}$ such that	
		\begin{equation}
			\label{step1alg}
			\begin{aligned}
				\control_e^{*d} \in \
				& \underset{\control_e \ge 0} {\argmax}
				& & a_e f_e^* \frac{y_e}{\frac{1}{\control_e}+\frac{r_e^{U_d}+r_e^{L_d}}{2a_e}} - \alpha h_e(\control_e).
			\end{aligned}
		\end{equation}}
	Find $e^{*d}$ such that	
	\begin{equation}
		\begin{aligned}
			e^{*d} \in \
			& \underset{e \in \mc{E}} {\argmax}
			& &  a_e f_e^*\frac{y_e}{\frac{1}{\control_e^{*d}}+\frac{r_e^{U_d}+r_e^{L_d}}{2a_e}} - \alpha h_e(\control_e^{*d}).
		\end{aligned}
		\label{step2alg}
	\end{equation}
	\caption{}
	\label{alg}
The optimal intervention is $\control^{*d} = \control_e^{*d} \delta^{(e^{*d})}$.
\end{algorithm}
Notice that the performance of Algorithm 1 depends on the choice of the parameter $d$. Specifically, the higher $d$ is the better is the approximation of the social cost variation.
\begin{theorem}
	\label{thm2}
	Let
	$\Delta C^{(\control)}$ be the social cost gain corresponding to intervention $\control = \control_e \delta^{(e)}$ as given in Theorem~\ref{thm}, and
	$$
	\Delta C_d^{(\control)}=a_ef_e^*\frac{y_e}{\frac{1}{\control_e}+\frac{r_e^{U_d}+r_{e}^{L_d}}{2a_e}}$$
	be the social cost gain estimated by Algorithm~\ref{alg} for a given distance $d\ge 1$. Then,
	\begin{equation*}
		\begin{aligned}
			\bigg|\frac{\Delta C^{(\control)}-\Delta C_d^{(\control)}}{\Delta C^{(\control)}}\bigg| 
			&\le \frac{\epsilon_{ed}}{2\Big(\frac{1}{\control_e}+\frac{r_e^{U_d}+r_e^{L_d}}{2a_e}\Big)} \\
			&\le \frac{\epsilon_{ed}}{2\Big(\frac{1}{\control_e}+\frac{1}{w^* \cdot a_e}\Big)},
			\label{err_rel2}
		\end{aligned}
	\end{equation*}
	where
	\begin{equation*}
		\label{epsilon}
		\epsilon_{ed}:=
		\frac{r_e^{U_d}-r_e^{L_d}}{a_e}.
	\end{equation*}
	Furthermore,
	\begin{equation}
		\Delta C^{(\control)} \ge a_e f_e^* \frac{y_e}{\frac{1}{\control_e}+\frac{r_e^{U_d}}{a_e}}.
		\label{min_gain}
	\end{equation}
\end{theorem}\medskip
\begin{proof}
See Appendix~\ref{app:proofs}.
\end{proof}

In the next section we provide sufficient conditions for $\epsilon_{ed}$ to vanish for large distance $d$ in the limit of infinite networks. In the rest of this section we show that the bound gap (and therefore $\epsilon_{ed}$) and the computational complexity of the bounds (for a fixed $d$) depend only on the local structure around link $M(e)$ of the resistor network, and do not scale with the network size, under the following assumption.
\begin{assumption}
	\label{sparse}
	Let $\mathcal{G}_R$ be the resistor network corresponding to the transportation network $\mc{G}$. Let $l \in \mc L$ be an arbitrary link of $\mc{G}_R$, and $\mathcal{N}_{\le d}$ denote the set of nodes that are at distance no greater than $d$ from link $l$.
	We assume that the network $\mc{G}$ is sparse in such a way that the cardinality of $\mathcal{N}_{\le d}$ does not depend on $\mathrm{N}$ for any $d$.
\end{assumption}

Assumption~\ref{sparse} is suitable for transportation networks, because of physical constraints not allowing for the degree of the nodes to grow unlimitedly (think for instance of planar grids, where the degree of the nodes is given no matter what the size of the network is, and the local structure of the network around an arbitrary node does not depend on the network size $\mathrm{N}$). Notice also that, under Assumption~\ref{sparse}, $\mathrm{N}$ and $\mathrm{L}$ are proportional. 
\begin{proposition}
	Let $\mathcal{G}_R=(\mc{N},\mc{L},W)$ be a resistor network, $l = \{i,j\}$ in $\mc{L}$, and $d \ge 1$. Then, $r_l^{U_d}$ and $r_l^{L_d}$, and their computational complexity, depend only on the structure of $\mathcal{G}_R$ within distance $d+1$ from $i$ and $j$ only. Furthermore, under Assumption~\ref{sparse} they do not depend on $\mathrm{N}$.
	\label{speed}
\end{proposition}
\begin{proof} 
See Appendix~\ref{app:proofs}.
\end{proof}

\begin{remark}
	To the best of our knowledge, the complexity of the most efficient algorithm to compute the spanning tree centrality (or effective resistance) of a link in large networks scales with the number of links \cite{hayashi2016efficient}. On the contrary,
	Proposition~\ref{speed} states that under Assumption~\ref{sparse} the computational time for approximating a single effective resistance does not scale with $\mathrm{N}$. Therefore, approximating all the effective resistances requires a computational time linear in $\mathrm{N}$. Observe that $v$ (and thus $y$) is computed via a diagonally dominant, symmetric and positive definite linear systems. The design of fast algorithms to solve this class of problem is an active field of research in the last years. To the best of our knowledge, the best algorithm has been provided in \cite{cohen2014solving} and has complexity
	$O(M \log^k \mathrm{N} \log \mathrm{1/\epsilon})$, where $\epsilon$ is the tolerance error, $k$ is a constant, and $M$ is the number of non-zero elements in the matrix of the linear system. Since in our case $M$ scales with $\mathrm{L}$, and since $\mathrm{L}$ scales with $\mathrm{N}$ under Assumption~\ref{sparse}, Algorithm~\ref{alg} is quasilinear in $\mathrm{N}$. Step \eqref{step1alg} consists in maximizing a function of one variable. Finally, step \eqref{step2alg} consists in taking the maximum of $\mathrm{E}$ numbers.
\end{remark}

\section{Bound analysis}
\label{performance}
In this section we characterize the gap between the bounds on the effective resistance of a link in terms of random walks over the resistor networks $\mathcal{G}_R$, $\mathcal{G}_{l}^{U_d}$ and $\mathcal{G}_{l}^{L_d}$. We then leverage this characterization to provide a sufficient condition on the network under which the bound gap vanishes asymptotically for large distance $d$. To this end, we interpret the conductance matrix $W$ of the resistor network as the transition rates of continuous-time Markov chain whose state space is the node set of the network, and introduce the following notation. Let:
\begin{itemize}
	\item $T_\mc{S}$ and $T_\mc{S}^+$ denote the hitting time (i.e., the first time $t\ge0$ such that the random walk hits the set $\mc{S} \in \mc{N}$), and the return time (i.e., the first time $t>0$ such that the random walk hits the set $\mc{S}$), respectively.
	\item $\mc{N}_d$ denote the set of the nodes that are at distance $d$ from link $l = \{i,j\}$, i.e., at distance $d$ from $i$ (or $j$) and at distance greater or equal than $d$ from $j$ (or $i$). Index $l$ is omitted for simplicity of notation.
	\item $p_k(X)$, $p_k^{U_d}(X)$ and $p_k^{L_d}(X)$, denote the probability that event $X$ occurs, conditioned on the fact that the random walk starts in $k$ at time $0$ and evolves over the resistor networks $\mathcal{G}_R$, $\mathcal{G}_{l}^{U_d}$ and $\mathcal{G}_{l}^{L_d}$, respectively.
\end{itemize} 
The next result provides a characterization of the bound gap in terms of random walks over $\mathcal{G}_R$, $\mathcal{G}_{l}^{U_d}$ and $\mathcal{G}_{l}^{L_d}$.
\begin{proposition}
	\label{gap}
	Let $\mathcal{G}_R=(\mathcal{N},\mathcal{L},W)$ be a resistor network. For every link $l = \{i,j\}$ in $\mc{L}$,
	\begin{equation}
		\begin{aligned}
		r^{U_d}_{l}-r^{L_d}_{l}
		& \le \frac{w_{i}}{(W_{ij})^2} \underbrace{p_i(T_{\mc{N}_d} < T_j)}_{\text{Term 1}} \ \cdot \\
		& \cdot \underset{g \in \mc{N}_d}{ \max} \underbrace{\big(p^{U_d}_g(T_i < T_j)-p^{L_d}_g(T_i < T_j)\big)}_{\text{Term 2}},
		\end{aligned}
		\label{equat_gap}
	\end{equation}
where the quantities in \eqref{equat_gap} are computed with respect to the continuous-time Markov chain with transition rates $W$. 
\end{proposition}
\begin{proof}
See Appendix~\ref{app:proofs}.
\end{proof}

In the next sections we shall use this result to analyze the asymptotic behaviour of the bound gap for an arbitrary link $l$ in $\mc{L}$ as $d \to +\infty$, for networks whose node set is infinite and countable. In particular, we show in Section~\ref{rec_sec} that this error vanishes asymptotically for the class of recurrent networks. The core idea to prove this result is to show that Term 1 vanishes. To generalize our analysis beyond recurrent networks, in Section~\ref{beyond} we study both Term 1 and 2 and provide examples showing that all combinations in Table~\ref{table_bounds} are possible. In particular, it is possible that the bound gap vanishes asymptotically for non-recurrent networks (for which Term 1 $\nrightarrow 0$, see \cite[Section 21.2]{levin2017markov}) if Term 2 $\rightarrow 0$.
\begin{table}
	\caption{All the four cases are possible, as shown in Section~\ref{beyond}. Term 1 $\rightarrow 0$ under the assumption that the network is recurrent, as proved in  Section~\ref{rec_sec}. \label{table_bounds}}
	\centering
		\begin{tabular}{lll}
			& Term 2 $\rightarrow 0$ & Term 2 $\nrightarrow 0$ \\
			\hline
			Term 1 $\rightarrow 0$
			& 2d grid
			& Ring \\
			Term 1 $\nrightarrow 0$
			& 3d grid
			& Double tree\\
			\hline
	\end{tabular}
\end{table}
\subsection{Recurrent networks} 
\label{rec_sec}
We start by introducing the class of recurrent networks.
\begin{definition}[Recurrent random walk]
	A random walk is recurrent if, for every starting point, it visits its starting node infinitely often with probability one \cite[Section 21.1]{levin2017markov}.
\end{definition}

\begin{definition}[Recurrent network]
	An infinite resistor network $\mathcal{G}_R=(\mathcal{N},\mathcal{L}, W)$ is recurrent if the random walk on the network is recurrent.
\end{definition}

The next theorem states that the bound gap vanishes asymptotically on recurrent networks if the degree of every node is finite. Note that the boundedness of the degree of all the nodes is guaranteed under Assumption~\ref{sparse}.
\begin{theorem}
	Let $\mathcal{G}_R=(\mc{N},\mc{L},W)$ be an infinite recurrent resistor network, and let $w^*<+\infty$. Then, for every $l$ in $\mc{L}$,
	\begin{equation*}
		\lim_{d \rightarrow +\infty} (r^{U_d}_{l}-r^{L_d}_{l})=0,
	\end{equation*}
	\label{recurrent}
\end{theorem}
\begin{proof}
It is proved in \cite[Proposition 21.3]{levin2017markov} that a network is recurrent if and only if 
\begin{equation}
	\lim_{d \rightarrow +\infty}p_i(T_{\mc{N}_d}<T_j)=0 \quad \forall l=\{i,j\} \in \mathcal{L}.
	\label{lim_rec2}
\end{equation}
Observe that, to hit any node in $\mc{N}_{d+1}$, the random walk starting from $i$ has to hit at least one node in $\mc{N}_d$. Hence, the sequence $\big\{p_i(T_{\mc{N}_d}<T_j)\big\}_{d=1}^\infty$ is non-increasing in $d$ and the limit in \eqref{lim_rec2} is well defined.
Then, from \eqref{equat_gap}, \eqref{lim_rec2}, from the fact that $0 \le p^{U_d}_g(T_i < T_j)-p^{L_d}_g(T_i < T_j)\le1$ for every node $g$, and from the assumptions $w^* < +\infty$ and $W_{ij}>0$ (recall that $i$ and $j$ are adjacent nodes), it follows
\begin{equation*}
	\lim_{d \rightarrow +\infty} r^{U_d}_{l}-r^{L_d}_{l}\le \frac{w^*}{(W_{ij})^2} \lim_{d \rightarrow +\infty} p_i(T_{\mc{N}_d}<T_j) = 0,
\end{equation*} 
which completes the proof.
\end{proof}
\begin{corollary}
Let $\mc{G}$ be a transportation network with recurrent associated resistor network $\mc{G}_R$. Then, for every $\control$ in $\Control$, 
$$
\lim_{d \to + \infty}\bigg|\frac{\Delta C^{(\control)}-\Delta C_d^{(\control)}}{\Delta C^{(\control)}}\bigg|=0.
$$
\end{corollary}\medskip
\begin{proof}
	The proof follows from Theorem~\ref{thm2} and Theorem~\ref{recurrent}, which imply $\lim_{d \rightarrow +\infty} \epsilon_{ed}=0$ for every $e$ in $\mc{E}$.
\end{proof}

Recurrence is a sufficient condition for the approximation error of a link effective resistance to vanish asymptotically, but is not necessary, as discussed in the next section.

\subsection{Beyond recurrence}
\label{beyond}
We here provide examples of infinite resistor networks for all of the cases in Table~\ref{table_bounds}.
Observe that, for every link $l = \{i,j\}$ in $\mathcal{L}$, the network cut at distance $d$ from $l$ and the network shorted at distance $d$ from $l$ differ for one node only (denoted by $s$), which is the result of shorting in a unique node all the nodes at distance greater than $d$ from $l$.
Intuitively speaking, our conjecture is that Term 2 in \eqref{equat_gap} is small when the network has many short paths. In this case, adding the node $s$ leads to a small variation of the probability, starting from any node in $\mc{N}_d$, of hitting $i$ before $j$, thus making Term 2 small. This intuition can be clarified with the next examples.
\subsubsection{2d grid}
Consider an infinite unweighted bidimensional grid as in Figure~\ref{quad3}. This network is very relevant for NDPs, since many transportation networks have similar topologies. The network is known to be recurrent \cite[Example 21.8]{levin2017markov}, hence Theorem~\ref{recurrent} guarantees that Term 1 vanishes asymptotically for every link $l = \{i,j\}$. Our conjecture, confirmed by numerical simulations, is that for every node $g$ in $\mc{N}_d$,
\begin{equation*}
	\begin{gathered}
		\lim_{d \rightarrow +\infty} p^{U_d}_g(T_i < T_j) =  1/2, \quad
		\lim_{d \rightarrow +\infty} p^{L_d}_g(T_i < T_j) =  1/2.
	\end{gathered}
\end{equation*} 
Hence, this is recurrent network for which also Term 2 vanishes asymptotically.
\subsubsection{3d grid}
Consider an infinite unweighted tridimensional grid. This network is not recurrent \cite[Example 21.9]{levin2017markov}, therefore Term 1 does not vanish asymptotically, and we cannot conclude from Theorem~\ref{recurrent} that for every $l = \{i,j\}$ the bound gap vanishes asymptotically. Nonetheless, numerical simulations show that, similarly to the bidimensional grid, for every node $g$ in $\mc{N}_d$,
\begin{equation*}
	\begin{gathered}
		\lim_{d \rightarrow +\infty} p^{U_d}_g(T_i < T_j) =  1/2, \quad
		\lim_{d \rightarrow +\infty} p^{L_d}_g(T_i < T_j) =  1/2.
	\end{gathered}
\end{equation*} 
Hence, this is a non-recurrent network for which Term 2 (and therefore the bound gap $r_l^{U_d}-r_l^{L_d}$) vanishes asymptotically in the limit of infinite distance $d$.
\subsubsection{Ring}
Consider an infinite unweighted ring network, and let us focus on nodes $5$ and $6$ in Figure~\ref{upper_lower_ring}. Then, 
\begin{equation*}
	\begin{gathered}
		p^{U_d}_{5}(T_1 < T_2)=1, \quad
		p^{U_d}_6(T_1 < T_2)=0.
	\end{gathered}
\end{equation*}
for each $d$ (even $d \rightarrow + \infty$), whereas, 
\begin{equation*}
	\begin{gathered}
		p^{L_d}_5(T_1 < T_2)=\frac{d}{2d+1} \xrightarrow[d \rightarrow + \infty]{} \frac{1}{2}, \\
		p^{L_d}_5(T_1 < T_2)=\frac{d+1}{2d+1} \xrightarrow[d \rightarrow + \infty]{} \frac{1}{2},
	\end{gathered}
\end{equation*}
since this case is equivalent to the gambler's ruin problem (see \cite[Proposition 2.1]{levin2017markov}).
Hence, Term 2 does not vanish for the ring. This is due to the fact that all the paths from $5$ to $2$ in $\mc G_l^{L_2}$ not including node $1$ include the node $s$. Still, Term 1 (and thus the bound gap $r_l^{U_d}-r_l^{L_d}$) vanishes asymptotically by Theorem~\ref{recurrent}, because the ring network is recurrent.

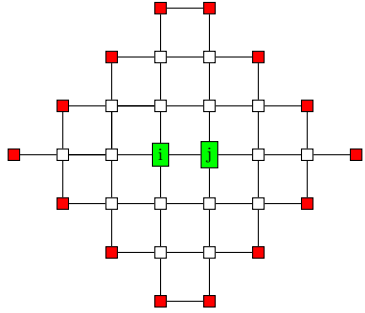
\begin{figure}
	\centering
		\begin{tikzpicture}[scale=0.65, transform shape]
			\node[draw,fill=green] (1) at (0,0)  {i};
			\node[draw,fill=green] (2) at (1,0)  {j};
			\node[draw,fill=none] (3) at (-1,0) {};
			\node[draw,fill=none] (4) at (0,1) {};
			\node[draw,fill=none] (5) at (1,1)  {};
			\node[draw,fill=none] (6) at (2,0)  {};
			\node[draw,fill=none] (7) at (1,-1)  {};
			\node[draw,fill=none] (8) at (0,-1)  {};
			\node[draw,fill=none] (9) at (-2,0) {};
			\node[draw,fill=none] (10) at (-1,1) {};
			\node[draw,fill=none] (11) at (0,2)  {};
			\node[draw,fill=none] (12) at (1,2)  {};
			\node[draw,fill=none] (13) at (2,1)  {};
			\node[draw,fill=none] (14) at (3,0)  {};
			\node[draw,fill=none] (15) at (2,-1)  {};
			\node[draw,fill=none] (16) at (1,-2)  {};
			\node[draw,fill=none] (17) at (0,-2)  {};
			\node[draw,fill=none] (18) at (-1,-1)  {};
			\node[draw,fill=red] (19) at (-3,0) {};
			\node[draw,fill=red] (20) at (-2,1) {};
			\node[draw,fill=red] (21) at (-1,2)  {};
			\node[draw,fill=red] (22) at (0,3)  {};
			\node[draw,fill=red] (23) at (1,3)  {};
			\node[draw,fill=red] (24) at (2,2)  {};
			\node[draw,fill=red] (25) at (3,1)  {};
			\node[draw,fill=red] (26) at (4,0)  {};
			\node[draw,fill=red] (27) at (3,-1)  {};
			\node[draw,fill=red] (28) at (2,-2)  {};
			\node[draw,fill=red] (29) at (1,-3) {};
			\node[draw,fill=red] (30) at (0,-3) {};
			\node[draw,fill=red] (31) at (-1,-2)  {};
			\node[draw,fill=red] (32) at (-2,-1)  {};
			\path (1) edge (2);
			\path (1) edge (3);
			\path (1) edge (4);
			\path (1) edge (8);
			\path (2) edge (5);
			\path (2) edge (6);
			\path (2) edge (7);
			\path (3) edge (9);
			\path (3) edge (10);
			\path (10) edge (4);
			\path (4) edge (11);
			\path (4) edge (5);
			\path (5) edge (12);
			\path (11) edge (12);
			\path (3) edge (9);
			\path (3) edge (10);
			\path (10) edge (4);
			\path (5) edge (13);
			\path (6) edge (13);
			\path (6) edge (14);
			\path (6) edge (15);
			\path (15) edge (7);
			\path (16) edge (7);
			\path (16) edge (17);
			\path (8) edge (17);
			\path (8) edge (18);
			\path (7) edge (8);
			\path (3) edge (18);
			\path (19) edge (9);
			\path (9) edge (20);
			\path (10) edge (20);
			\path (10) edge (21);
			\path (11) edge (21);
			\path (11) edge (22);
			\path (22) edge (23);
			\path (12) edge (23);
			\path (12) edge (24);
			\path (24) edge (13);
			\path (25) edge (13);
			\path (25) edge (14);
			\path (26) edge (14);
			\path (27) edge (14);
			\path (15) edge (28);
			\path (15) edge (27);
			\path (16) edge (28);
			\path (16) edge (29);
			\path (29) edge (30);
			\path (17) edge (30);
			\path (17) edge (31);
			\path (31) edge (18);
			\path (32) edge (9);
			\path (32) edge (18);
	\end{tikzpicture}
	\caption{Bidimensional square grid, cut at distance $d=3$ from $l=\{i,j\}$. The red nodes belong to $\mc{N}_d$. As $d$ increases, $p_g(T_1 < T_2)$ approaches $1/2$ for every node $g$ in $\mc{N}_d$. \label{quad3}}
\end{figure}
\begin{figure}
			\centering{\hspace{0.6cm} $\mathcal{G}_{l}^{U_2}$ \hspace{1.8cm} $\mathcal{G}_{l}^{L_2}$}\\
			\vspace{0.2cm}
			\begin{tikzpicture}[scale=0.75, transform shape]
				\node[draw, circle, fill=green] (1) at (-1,0)  {$1$};
				\node[draw, circle, fill=green] (2) at (1,0) {$2$};
				\node[draw, circle] (3) at (-1,1)  {$3$};
				\node[draw, circle] (4) at (1,1) {$4$};
				\node[draw, circle, fill=red] (5) at (-1,2) {$5$};
				\node[draw, circle, fill=red] (6) at (1,2) {$6$};;
				\path   (1) edge [bend right=0]
				node [above] {$l$} (2);
				\path   (1) edge [bend right=0]
				node [above] {} (3);
				\path   (2) edge [bend left=0]
				node [above] {} (4);
				\path (3) edge [bend right=0]
				node [above] {} (5);
				\path   (4) edge [bend right=0]
				node [below] {} (6);
			\end{tikzpicture}
			\qquad
			\begin{tikzpicture}[scale=0.56, transform shape]
				\node[draw, circle, fill=green] (1) at (-1,0)  {$1$};
				\node[draw, circle, fill=green] (2) at (1,0) {$2$};
				\node[draw, circle] (3) at (-1,1)  {$3$};
				\node[draw, circle] (4) at (1,1) {$4$};
				\node[draw, circle, fill=red] (5) at (-1,2) {$5$};
				\node[draw, circle, fill=red] (6) at (1,2) {$6$};
				\node[draw, circle] (7) at (0,3) {$s$};
				\path   (1) edge [bend right=0]
				node [above] {$l$} (2);
				\path   (1) edge [bend right=0]
				node [above] {} (3);
				\path   (2) edge [bend left=0]
				node [above] {} (4);
				\path (3) edge [bend right=0]
				node [above] {} (5);
				\path   (4) edge [bend right=0]
				node [below] {} (6);
				\path   (5) edge [bend left=0]
				node [below] {} (7);
				\path  (6) edge [bend left=0]
				node [below] {} (7);
			\end{tikzpicture}
	\caption{\emph{Left}: ring cut at distance $d=2$ from $l$. \emph{Right}: ring shorted at distance $d=2$ from link $l=\{1,2\}$.} \label{upper_lower_ring}
\end{figure}
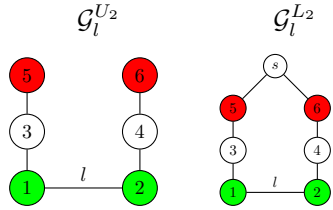
\subsubsection{Double tree network}
\label{section_tree}
The last examples illustrates an infinitely large network in which the bound gap does not vanish asymptotically. This network is not relevant for traffic applications, since it admits one path only between every pair of nodes, but provides an interesting counterexample where the bound gap does not converge asymptotically.
The network is composed of two infinite trees starting from node $i$ and $j$, connected by a link $l = \{i,j\}$ (see Figure~\ref{counterex}), and is unweighted. It can be shown that on this network the probability that a random walk, starting from $i$, returns on $i$ is equal to the same quantity computed on a biased random walk over an infinite line (for more details see Appendix~\ref{app:tree}). Since the biased random walk on a line is not recurrent \cite[Example 21.2]{levin2017markov}, then the double tree network is non-recurrent, and Term 1 $\nrightarrow 0$. Moreover, we show in Appendix~\ref{app:tree} that
\begin{equation*}
	\lim_{d \rightarrow +\infty} r^{U_d}_{l}-r^{L_d}_{l}=\frac{1}{3},
\end{equation*}
thus implying that Term 2 $\nrightarrow 0$.
\begin{figure}
	\centering
	\begin{tikzpicture}[scale=0.65, transform shape]
		\node[draw, circle, fill=green] (1) at (-2,0)  {$i$};
		\node[draw, circle, fill=green] (2) at (2,0) {$j$};
		\node[draw, circle] (3) at (-2.8,1)  {};
		\node[draw, circle] (4) at (-1.2,1) {};
		\node[draw, circle] (5) at (1.2,1) {};
		\node[draw, circle] (6) at (2.8,1) {};
		\node[draw, circle] (7) at (-3.2,2)  {};
		\node[draw, circle] (8) at (-2.4,2) {};
		\node[draw, circle] (9) at (-1.6,2) {};
		\node[draw, circle] (10) at (-0.8,2) {};
		\node[draw, circle] (11) at (0.8,2)  {};
		\node[draw, circle] (12) at (1.6,2) {};
		\node[draw, circle] (13) at (2.4,2) {};
		\node[draw, circle] (14) at (3.2,2) {};
		\node[draw=none] (15) at (-3.5,3)  {};
		\node[draw=none] (16) at (-2.9,3) {};
		\node[draw=none] (17) at (-2.7,3) {};
		\node[draw=none] (18) at (-2.1,3) {};
		\node[draw=none] (19) at (-1.9,3)  {};
		\node[draw=none] (20) at (-1.3,3) {};
		\node[draw=none] (21) at (-1.1,3) {};
		\node[draw=none] (22) at (-0.5,3) {};
		\node[draw=none] (23) at (0.5,3)  {};
		\node[draw=none] (24) at (1.1,3) {};
		\node[draw=none] (25) at (1.3,3) {};
		\node[draw=none] (26) at (1.9,3) {};
		\node[draw=none] (27) at (2.1,3)  {};
		\node[draw=none] (28) at (2.7,3) {};
		\node[draw=none] (29) at (2.9,3) {};
		\node[draw=none] (30) at (3.5,3) {};
		\path   (1) edge [bend right=0]
		node [above] {$l$} (2);
		\path   (1) edge [bend right=0]
		node [above] {} (3);
		\path   (1) edge [bend left=0]
		node [above] {} (4);
		\path (2) edge [bend right=0]
		node [above] {} (5);
		\path   (2) edge [bend right=0]
		node [below] {} (6);
		\path   (3) edge [bend left=0]
		node [below] {} (7);
		\path  (3) edge [bend left=0]
		node [below] {} (8);
		\path   (4) edge [bend right=0]
		node [above] {} (9);
		\path   (4) edge [bend right=0]
		node [above] {} (10);
		\path   (5) edge [bend left=0]
		node [above] {} (11);
		\path (5) edge [bend right=0]
		node [above] {} (12);
		\path   (6) edge [bend right=0]
		node [below] {} (13);
		\path   (6) edge [bend left=0]
		node [below] {} (14);
		\path   (7) edge [dashed]
		node [above] {} (15);
		\path   (7) edge [dashed]
		node [above] {} (16);
		\path   (8) edge [dashed]
		node [above] {} (17);
		\path (8) edge [dashed]
		node [above] {} (18);
		\path   (9) edge [dashed]
		node [below] {} (19);
		\path   (9) edge [dashed]
		node [below] {} (20);
		\path  (10) edge [dashed]
		node [below] {} (21);
		\path   (10) edge [dashed]
		node [above] {} (22);
		\path   (11) edge [dashed]
		node [above] {} (23);
		\path   (11) edge [dashed]
		node [above] {} (24);
		\path (12) edge [dashed]
		node [above] {} (25);
		\path   (12) edge [dashed]
		node [below] {} (26);
		\path   (13) edge [dashed]
		node [below] {} (27);
		\path (13) edge [dashed]
		node [above] {} (28);
		\path   (14) edge [dashed]
		node [below] {} (29);
		\path   (14) edge [dashed]
		node [below] {} (30);
	\end{tikzpicture}
	\caption{The double tree is an infinite non-recurrent network. 
		\label{counterex}}
\end{figure}
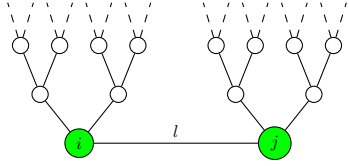

\section{Numerical simulations}
\label{sec:simulations}
This section is devoted to numerical simulations.
In Section~\ref{sec:res_approx} we analyze the bound gap for finite distance $d$, both on real and synthetic transportation networks.
Then, we discuss in Section~\ref{sec:relax_ass} how to adapt our method to more general
NDPs with non-linear delay functions, and provide numerical simulations showing that our algorithm may be applied in real scenarios even if the regularity assumption on the Wardrop equilibrium (i.e., Assumption \ref{assumption}) is violated.

\subsection{Effective resistance approximation}
\label{sec:res_approx}
\subsubsection{Infinite grids}
Infinite regular grids are relevant networks to test the performance of the bounds on the effective resistance, since they are good proxy for transportation networks.
In Table~\ref{table} the bound gap in a square grid network with unitary conductances is shown. Similar results are obtained in any regular infinite grid. 
Numerical simulations show that for every link $l$ in $\mc{L}$,
\begin{equation*}
	\frac{r^{U_d}_{l}-r_{l}}{r_{l}}=\frac{r_{l}-r^{L_d}_{l}}{r_{l}}=O(1/d^2).
\end{equation*}
We emphasize that, despite the network being infinitely large, even at $d=5$ the bounds are close to the true value effective resistance, which is 1/2 \cite{bartis1967let}.
\begin{table}[]
	\caption{Table of upper and lower bound in infinite square grid. \label{table}}
	\centering
	\begin{tabular}{llllll}
		& $d=1$ & $d=2$ & $d=3$ & $d=4$ & $d=5$\\
		\hline
		$(r_{l}^{U_d}-r_{l})/r_{l}$ & $1/5$ & $0.0804$ & $0.0426$ & $0.0262$ & $0.0178$\\
		\hline
		$(r_{l}-r_{l}^{L_d})/r_{l}$ & $1/5$ & $0.0804$ & $0.0426$ & $0.0262$ & $0.0178$\\
		\hline
	\end{tabular}
\end{table}

\subsubsection{Oldenburg
	transportation network}
\begin{figure}
	\centering
	\includegraphics[width=5cm]{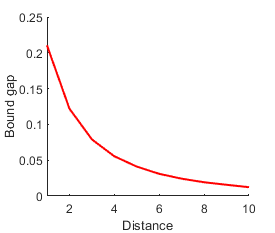}
	\caption{Average relative gap of the bounds on Oldenburg network as a function of distance $d$. \label{old_fig}}
\end{figure}
In this section we illustrate the performance of our bounds on the effective resistance of links of the resistor network associated to the transportation network of Oldenburg \cite{brinkhoff2002framework}. The transportation network is composed of $6105$ nodes and $7035$ links, and the diameter of the associated resistor network (i.e., maximum distance between pair of nodes) is $104$. We assume for simplicity $a_e=1$ for every link $e \in \mc{E}$, but numerical results prove to be robust with respect to some variability in those parameters. 
The average relative bound gap on the associated resistor network, defined as
\begin{equation*}
	\Delta R_d:=\frac{1}{\mathrm{L}}\sum_{l \in \mathcal{L}}\frac{r^{U_d}_{l}-r^{L_d}_{l}}{r_{l}}
\end{equation*} is shown in Table~\ref{oldenburg_lower_bound} and Figure~\ref{old_fig}. 
\begin{table}[]
	\caption{Table of average relative error bound gap at distance $d$ on the Oldenburg network. \label{oldenburg_lower_bound}}
	\centering
		\begin{tabular}{llllllll}
			& d=1 & d=2 & d=3 & d=4 & d=5 & d=6 & d=7\\
			\hline
			$\Delta R_d$
			& 0.21
			& 0.12
			& 0.079
			& 0.056
			& 0.041
			& 0.031
			& 0.024\\
			\hline
	\end{tabular}
\end{table}
We observe that also in this network the bound gap decreases quickly compared to the diameter of the network.

\subsection{Relaxing assumptions}
\label{sec:relax_ass}
The goal of this section is two-fold. We first show how to adapt Theorem \ref{thm} when the delay functions are non-affine, and validate by numerical analysis the proposed method. We then show that violating Assumption \ref{assumption} is not a practical issue in real case scenarios. The numerical example is based on the highway network of Los Angeles (see Figure~\ref{fig:mapLA} \cite{los_angeles}). To handle non-linear delay functions, the main idea is to adapt Theorem~\ref{thm} by constructing a resistor network and then follow same steps as in Algorithm~\ref{alg}. To this end, let us write the KKT conditions of \eqref{convex_prob} as follows:
$$
\left[ 
\begin{array}{ c  c } 
	\text{diag}\Big(\Big\{\frac{\tau_e(f_e^*)}{f_e^*}\Big\}_{e \in \mc E}\Big) & -(B_-)^T \\ 
	-B_- & \mathbb{0} 
\end{array} 
\right]
\left[
\begin{array}{c}
	f^* \\ \gamma^*_-
\end{array}
\right]
=
-\left[
\begin{array}{c}
	\tau_e(0) \\ \nu_-
\end{array}
\right],
$$
where $f^*$ and $\gamma^*$ denote the Wardrop equilibrium and the optimal Lagrangian multipliers before the intervention. The KKT conditions suggest that in non-affine routing games the term $\tau_e(f_e^*)/f_e^*$ plays the role of $a_e$ in affine routing games (see the proof of Theorem~\ref{thm} in Appendix~\ref{app:proofs} for more details). Hence, by following similar steps as in affine routing games, we construct a resistor network with conductance matrix
\begin{equation}
	\label{eq:new_W}
W_{ij}= \begin{cases}
	\sum_{\substack{e \in \mathcal{E}: \\
			\xi(e)=i, \theta(e)=j, \ \text{or} \\
			\xi(e)=j, \theta(e)=i}} \frac{f_e^*}{\tau_e(f_e^*)} & \text{if}\ i \neq j \\
	0 & \text{if}\ i=j. 
\end{cases}
\end{equation}
The social cost variation for single-link interventions is then computed by using Theorem \ref{thm} with respect to the new resistor network with conductance matrix \eqref{eq:new_W}. Observe that, in contrast with the affine case, this method is not exact for non-linear delay functions, since the Wardrop equilibrium (and thus the elements of $W$) are modified by interventions, not allowing to leverage Sherman-Morrison theorem to compute the social cost variation.

To validate our method we assume that delay functions are in the form $\tau_e(f_e) = a_e (f_e)^4+b_e$, and consider interventions in the form $\control = 3 \delta^{(e)}$ for every $e$ in $\mc{E}$. Numerical parameters are not reported in the paper due to limited space, but the obtained results are robust with respect to a change of numerical values. For every intervention, we compare the social cost variation computed by two methods: \emph(i) by solving the convex optimization~\eqref{eq:f(k)} and plugging the new equilibrium $f^*(\control)$ into the social cost function (\emph{exact}); \emph(ii) via the electrical formulation, i.e., by leveraging Theorem~\ref{thm} with conductance matrix \eqref{eq:new_W} and ignoring the fact that Assumption~\ref{assumption} may be violated (\emph{approximated}). Figure~\ref{fig:non_linear} illustrates the social cost variation computed by the two methods corresponding to interventions on the five links of the network that yield the largest cost variation. The numerical simulations show that support of the equilibrium varies with the intervention. Nonetheless, the proposed method approximates quite well the social cost variation and selects the optimal link for the intervention. The implication of combining the results of this section and Section \ref{sec:res_approx} is that Algorithm~\ref{assumption} should manage to select optimal (or weakly suboptimal) interventions in large transportation networks also when the delay functions are non-linear, Assumption~\ref{assumption} is violated, and effective resistances are computed at small distance~$d$.
\begin{figure}[t!]
	\centering
	\includegraphics[scale=0.22]{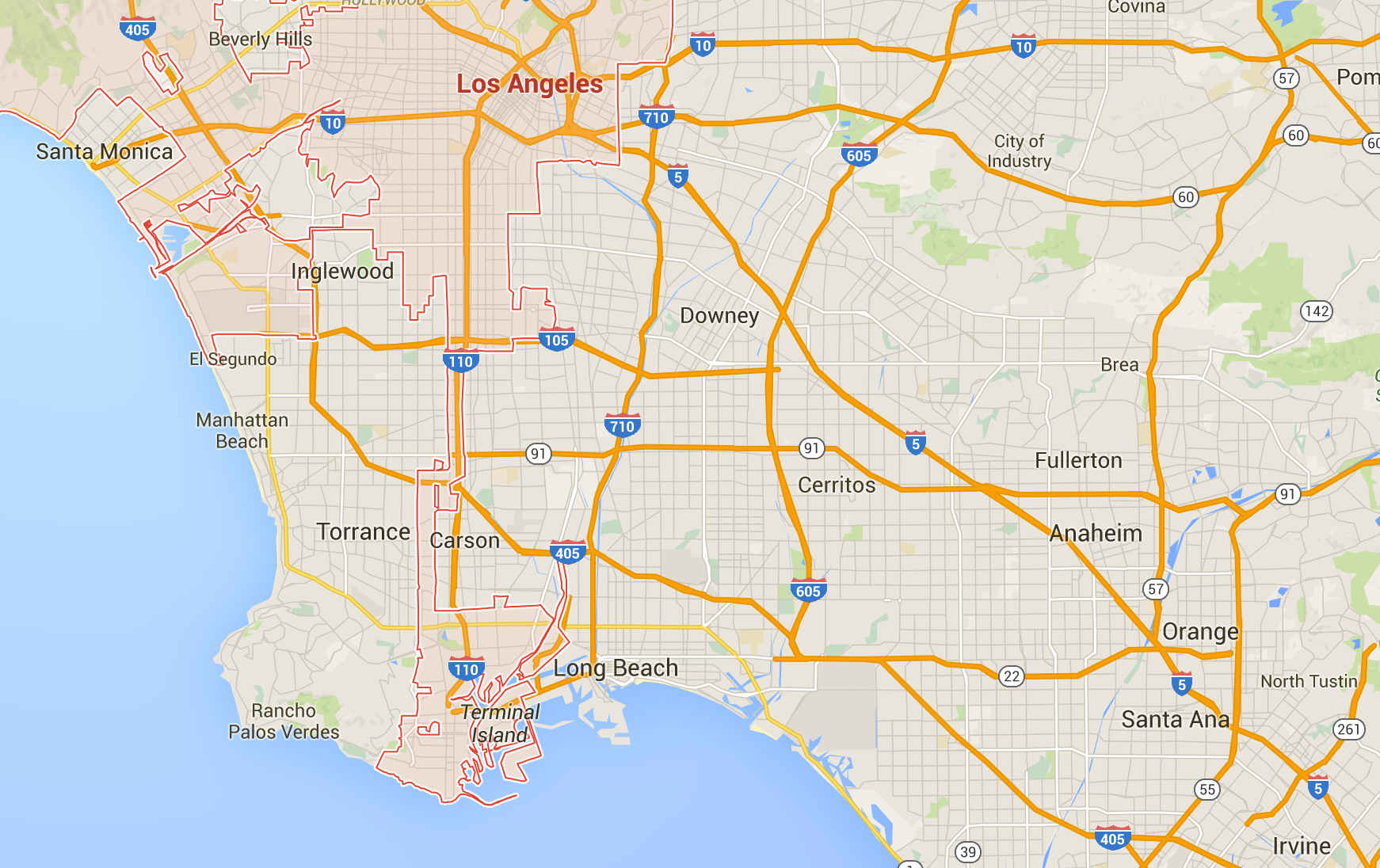}\\[12pt]
	\begin{tikzpicture}[scale=0.55,transform shape]
		\node[draw, circle] (1) at (1,0) {1}; 
		\node[draw, circle] (2) at (4,0) {2}; 
		\node[draw, circle] (3) at (6,0) {3}; 
		\node[draw, circle] (4) at (10,0) {4};
		\node[draw, circle] (5) at (14,-1) {5}; 
		\node[draw, circle] (6) at (2,-4) {6};  
		\node[draw, circle] (7) at (4,-4) {7};
		\node[draw, circle] (8) at (6,-4) {8};
		\node[draw, circle] (9) at (8,-3) {9}; 
		\node[draw, circle] (10) at (4,-8) {10}; 
		\node[draw, circle] (11) at (6,-8) {11}; 
		\node[draw, circle] (12) at (8, -8) {12}; 
		\node[draw, circle] (13) at  (12, -9) {13}; 
		\node[draw, circle] (14) at (14, -9) {14}; 
		\node[draw, circle] (15) at (6, -10) {15}; 
		\node[draw, circle] (16) at (9,-10.5) {16};
		\node[draw, circle] (17) at  (14, -11)  {17}; 
		\path[->,shorten >=1pt,auto, node distance = 0.5cm, semithick]
		(1) edge node {$l_1$} (2) 
		(2) edge node {$l_2$}  (3)
		(3) edge node {$l_3$} (4)
		(4) edge node {$l_4$} (5)
		(1) edge node {$l_5$} (6)
		(6) edge node {$l_6$} (7)
		(7) edge node {$l_7$} (8)
		(8) edge node {$l_8$} (9)
		(9) edge node {$l_{9}$} (13)
		(2) edge node {$l_{10}$} (7)
		(3) edge node {$l_{11}$} (8)
		(3) edge node {$l_{12}$} (9)
		(4) edge node {$l_{13}$} (9)
		(5) edge node {$l_{14}$} (14)
		(6) edge node {$l_{15}$} (10)
		(10) edge node {$l_{16}$} (11)
		(10) edge[left] node {$l_{17}$} (15)
		(7) edge node {$l_{18}$} (10)
		(8) edge node {$l_{19}$} (11)
		(9) edge node {$l_{20}$} (12)
		(11) edge node {$l_{21}$} (12)
		(12) edge node {$l_{22}$} (13)
		(13) edge node {$l_{23}$} (14)
		(11) edge node {$l_{24}$} (15)
		(13) edge[left] node {$l_{25}$} (17)
		(14) edge node {$l_{26}$} (17)
		(15) edge[below] node {$l_{27}$} (16)
		(16) edge[below] node {$l_{28}$} (17);
	\end{tikzpicture}
	\caption{\emph{Top}: the highway network in Los Angeles. \emph{Bottom}: a graph representation of the network, where node $1$ (Santa Monica) and $17$ (Santa Ana) are respectively the origin and the destination. \label{fig:mapLA}}
\end{figure} 
\begin{figure}
	\centering
	\includegraphics[width=7cm]{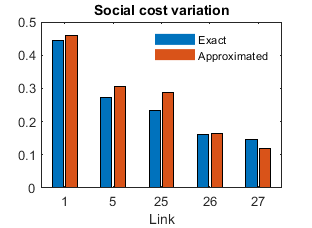}
	\caption{\emph{Top}: Social cost variation for interventions in the form $\control = 3 \delta^{(e)}$ for a routing game on the graph of Figure~\ref{fig:mapLA} with delay functions in the form $\tau_e(f_e) = a_e (f_e)^4 + b_e$. The cost variation is computed by solving convex optimization (\emph{exact}) and by adapting Theorem~\ref{thm} to the case of non-linear delay functions (\emph{approximated}), as explained in Section~\ref{sec:relax_ass}. The plot illustrates the social cost variation for the five links that maximize the cost variation.\label{fig:non_linear}}
\end{figure}

\section{Conclusion}
In this work we study a network design problem where a single link can be improved. Under the assumption that the support of the Wardrop equilibrium is not modified with an intervention, we reformulate the problem in terms of electrical quantities computed on a related resistor network, in particular in terms of the effective resistance of a link. We then provide a method to approximate such an effective resistance by performing only local computation, which may be of separate interest. 
Based on the electrical formulation and our approximation method for the effective resistance we propose an efficient algorithm to solve efficiently the network design problem. We then show by numerical examples that our method can be adapted to routing games with non-linear delay functions, and achieves good performance even if the support of the equilibrium is modified by the intervention.

An interesting direction for the future is a deeper analysis on tightness of the bounds on effective resistance for finite distance $d$. Future research lines also include extending the analysis to the case of multiple interventions. Indeed, the general problem is not submodular, thus guarantees on the performance of greedy algorithm are not given. A possible direction is to exploit the closed formula for the social cost derivative to implement gradient descents algorithms. Other directions include extending the theoretical framework to the case of multiple origin-destination pairs and heterogeneous preferences \cite{Cianfanelli.ea:19,Cianfanelli.ea:22}.

\bibliographystyle{plain}
\bibliography{bibliography.bib}

\appendix

\section{Preliminaries on connection between Green's function, random walks and effective resistance}
\label{preliminaries}
Let $\mathcal{G}_R=(\mathcal{N},\mathcal{L},W)$ denote a connected resistor network (this is without loss of generality for resistor network associated to transportation networks), and $P=\mb{I}_w^{-1}W$ the transition probability matrix of the jump chain of the continuous-time Markov chain with rates $W$. We denote by $_kP$ the matrix obtained by deleting from $P$ the row and the column referring to the node $k$. $_kP$ can be thought of as the transition matrix of a killed random walk obtained by creating a cemetery in the node $k$. We then define the Green's function 
as
\begin{equation}
	_kG:=\sum_{t=0}^{\infty} (_kP)^t=(\mathds{I}-{_kP})^{-1}.
	\label{G}
\end{equation}
The last inequality in \eqref{G} follows from the connectedness of $\mathcal{G}_R$, which implies that $_kP$ is substochastic and irreducible. Hence, it has spectral radius and the inversion is well defined \cite{horn2012matrix}.
Since $((_kP)^t)_{ij}$ is the probability that the killed random walk starting from $i$ is in $j$ after $t$ steps, $_kG_{ij}$ indicates the expected number of times that the killed random walk visits $j$ starting from $i$ before being absorbed in $k$ \cite{ellens2011effective}.
It is known that the Green's function of the random walk on a resistor network can be related to electrical quantities \cite{ellens2011effective}. In particular, with the convention that
\begin{equation}
	_kG_{ik}={_kG_{ki}}={_kG_{kk}}=0 \quad \forall i \in \mathcal{N},
	\label{green0}
\end{equation}
it is known that for any node $k$ and link $l = \{i,j\}$ in $\mc L$,
\begin{equation}
	\begin{aligned}
	r_{l}&=
	\frac{_kG_{ii}-{_kG}_{ji}}{w_{i}}+\frac{_kG_{jj}-{_kG}_{ij}}{w_{j}}\\
	&=\frac{1}{w_{i}p_i(T_j<T_i^+)}=\frac{_jG_{ii}}{w_{i}},
	\label{reff}
	\end{aligned}
\end{equation}
where $p_i(T_j<T_i^+)$ is defined in Section~\ref{performance}, and $r_{l}$ is the effective resistance of link $l$ as defined in Definition~\ref{def:effective}.

\section{Proofs}
\label{app:proofs}

\subsection{Proof of Lemma~\ref{cost}}
From \eqref{kkt2}, for all the used links $e$ in $\mc{E}$,
\begin{equation*}
	\gamma_{\xi(e)}^* - \gamma_{\theta(e)}^* = \tau_e(f_e^*).
\end{equation*}
Consider a path $p=(e_1,e_2,...e_s)$, with $\xi(e_1)=\mathrm{o}$, $\theta(e_s)=\mathrm{d}$, and $\theta(e_i)=\xi(e_{i+1})$ for every $1 \le i < s$. Thus, from \eqref{cost_path},
\begin{equation*}
	c_p(f^*)=\sum_{i=1}^s \tau_{e_i}(f_{e_i}^*)
	=\sum_{i=1}^{s}( \gamma^*_{\xi(e_i)}-\gamma^*_{\theta(e_{i})})
	=\gamma_\mathrm{o}^*-\gamma_\mathrm{d}^*.
\end{equation*}
Hence, all the used paths at the equilibrium have the same cost $\gamma_\mathrm{o}^*-\gamma_\mathrm{d}^*$. Then, the social cost is
\begin{equation*}
	\begin{split}
		C^{(0)}&=\sum_{e \in \mathcal{E}} f^*_e \tau_e(f^*_e)=\sum_{e \in \mathcal{E}} \tau_e(f_e^*)\sum_{p \in \mathcal{P}} A_{ep}z^*_p\\
		&=\sum_{p \in \mathcal{P}}z^*_p \sum_{e \in \mathcal{E}}A_{ep}\tau_e(f_e^*)=
		\sum_{p \in \mathcal{P}}z^*_p c_p(f^*)\\
		&=(\gamma_\mathrm{o}^*-\gamma_\mathrm{d}^*)\sum_{p \in \mathcal{P}} z^*_p = m (\gamma_\mathrm{o}^*-\gamma_\mathrm{d}^*),
	\end{split}
\end{equation*} 
where the second equivalence follows from \eqref{incidence_path}, the fourth one from \eqref{cost_path}, 
and the last one from \eqref{constraints}.

\subsection{Proof of Theorem~\ref{thm}}
Consider the KKT conditions \eqref{kkt2}, and let us remove the links in $\mathcal{E}_+$. Thus, the last three conditions of \eqref{kkt2} can be ignored without affecting the solution. With a slight abuse of notation, from now on let $\mathcal{E}$ denote $\mathcal{E} \setminus \mathcal{E}_+$.
Using the fact that the delay functions are affine, the KKT conditions become:
\begin{align*}
	\begin{cases}
		a_e f_e^*+b_e+\gamma_{\theta(e)}^*-\gamma^*_{\xi(e)} =0 & \quad \forall e \in \mathcal{E},\\
		\sum_{e \in \mathcal{E}: \theta(e)=i} f_e^* - \sum_{e \in \mathcal{E}: \xi(e)=i} f_e^* + \nu_i=0 & \quad \forall i \in \mathcal{N},
	\end{cases}
	\label{kkt}
\end{align*}
where the constraint $f^*_{e}\ge 0$ can now be removed since the solution of the new KKT conditions gives $f^*_{e}\ge0$ for every link $e$ not in $\mathcal{E}_+$.
Observe that the optimal flow $f^*_{e}$ depends on $\gamma^*$ only via the difference $\gamma^*_{\xi(e)} - \gamma^*_{\theta(e)}$, so that $\gamma^*$ remains a solution if a constant vector is added to it. This is due to the fact that the matrix $B$ is not full rank. Observe that removing the last row of $B$ is equivalent to imposing $\gamma^*_\mathrm{d}=0$.
We let $\gamma_-$ and $\nu_-$ denote respectively $\gamma$ and $\nu$ where the last element of both vectors is removed, and let $B_- \in \mathds{R}^{(\mathrm{N}-1) \times \mathrm{E}}$ denote the node-link incidence matrix where the last row is removed. Finally, we define $H \in \mathds{R}^{({\mathrm{N}+\mathrm{E}-1})\times({\mathrm{N}+\mathrm{E}-1})}$ as
\begin{equation*}
	H:=
	\left[ 
	\begin{array}{ c  c } 
		\mb{I}_a & -(B_-)^T \\ 
		
		-B_- & \mathbb{0} 
	\end{array} 
	\right].
	\label{H}
\end{equation*}
With this notation in mind, and assuming $\gamma_\mathrm{d}^*=0$, the KKT conditions may be written in compact form as
\begin{equation}
	H\left[
	\begin{array}{c}
		f^* \\ \gamma_-^*
	\end{array}
	\right]=-\left[
	\begin{array}{c}
		b \\ \nu_-
	\end{array}
	\right].    
	\label{compact_kkt2}
\end{equation}
Since we take $\gamma^*_\mathrm{d}=0$, the system has unique solution, i.e.,
\begin{equation}
	\left[
	\begin{array}{c}
		f^* \\ \gamma_-^*
	\end{array}
	\right]=\left[\begin{array}{ c c } 
		KQ^{-1}K^T-\mb{I}_a^{-1} & KQ^{-1} \\ 
		
		Q^{-1}K^T & Q^{-1}
	\end{array}\right]\left[
\begin{array}{c}
b \\ \nu_-
\end{array}
\right],
	\label{linear}
\end{equation}
where $K:=\mb{I}_a^{-1}B_-^T \in \mathds{R}^{\mathrm{E} \times (\mathrm{N}-1)}$ and $Q:=B_-\mb{I}_a^{-1}B_-^T \in \mathds{R}^{(\mathrm{N}-1) \times (\mathrm{N}-1)}$.
The invertibility of $H$ follows from the invertibility of $\mb{I}_a$ (the delays are strictly increasing) and from the invertibility of $Q$ (see \cite{horn2012matrix}), which will be proved in a few lines.
From the definitions of $B_-$ and $a$, it follows that for every link $e$,
\begin{equation} 
	K_{e:}=\frac{(\delta^{(\xi(e))})^T-(\delta^{(\theta(e))})^T}{a_e},
	\label{K}
\end{equation}
with the convention that $\delta^{(\mathrm{d})}=0\cdot \mathbf{1}$ (since we removed the destination in $B_-$). Moreover,
\begin{equation*}
	Q_{ij}= \begin{cases}
		-\sum_{\substack{e \in \mathcal{E}: \\
				\xi(e)=i, \theta(e)=j, \ \text{or} \\
				\xi(e)=j, \theta(e)=i}}\frac{1}{a_e} & \text{if}\ i \neq j\\
		\sum_{l \in \partial i} \frac{1}{a_e} & \text{if}\ i = j.
	\end{cases} \quad \forall i,j \in \mathcal{N} \setminus \mathrm{d},
	\label{Q}
\end{equation*}
where $\partial i$ denotes the in and out neighborhood links of $i$, i.e.,
\begin{equation*}
	\partial i := \{e \in \mathcal{E}: B_{ie}\neq 0 \}.
\end{equation*} 
Let $L = I_{w} - W$ denote the Laplacian of the associated resistor network $\mc G _R$, and $_{\dd}L$ denote its restriction to $\mc{N} \setminus \dd$. We remark that $\partial i$ includes also links pointing to the destination. This allows to observe that $Q={_\dd L}$, which implies the invertibility of $Q$.
Let $\mb{I}_a^{(\control)}$, $H^{(\control)}, Q^{(\control)}$ and $K^{(\control)}$ denote the matrix $\mb{I}_a, H, Q$ and $K$ corresponding to the intervention $\control$.
Note that an intervention on link $e$ corresponds to a rank-1 perturbation of $Q$. In particular,
	\begin{equation*}
		Q^{(\control_e \delta^{(e)})} = Q + \frac{\control_e}{a_e} B_{-}^e (B_{-}^e)^T,
	\end{equation*}
	where $B_{-}^e$ denotes the $e-$th column of $B_-$. Thus, by Sherman-Morrison formula,
	\begin{equation}
		\label{eq:Qe}
		(Q^{(\control_e \delta^{(e)})})^{-1} = Q^{-1} - \frac{Q^{-1}B_-^e (B_{-}^e)^T Q^{-1}}{\frac{a_e}{\control_e}+(B_{-}^e)^T Q^{-1} B_{-}^e}.
	\end{equation}
	Let for simplicity of notation assume $\xi(e)=i, \theta(e)=j$. Then,
	\begin{equation}
		\label{eq:Ke}
		K^{(\control_e \delta^{(e)})}-K = \frac{\control_e}{a_e} \delta^{(e)}(\delta^{(i)}-\delta^{(j)})^T= \frac{\control_e}{a_e} \delta^{(e)}(B_-^e)^T.
	\end{equation}
	By \eqref{compact_kkt2}, \eqref{eq:Qe}, and \eqref{eq:Ke}, we thus get
	\begin{equation}
		\label{eq:delta_gamma}
		\begin{aligned}
			\gamma^*_{\mathrm{o}} - \gamma^*_{\mathrm{o}}(\control_e \delta^{(e)}) &=-\frac{\control_e}{a_e} Q^{-1}B_-^e(\delta^{(e)})^T b +\\
			&+ \frac{Q^{-1}B_-^e (B_{-}^e)^T Q^{-1}}{\frac{a_e}{\control_e}+(B_{-}^e)^T Q^{-1} B_{-}^e} \times \\
			&\times \left(K^T b + \frac{\control_e}{a_e}B_-^e(\delta^{(e)})^T b +\nu_- \right).
		\end{aligned}
	\end{equation}
	We now give an interpretation to the terms in equation \eqref{eq:delta_gamma}. Let $\tilde{\mb{I}}_w$ and $\tilde{W}$ denote the restriction of $\mb{I}_w$ and $W$ over $\mc{N} \setminus \dd$, respectively.
	Note that $ _\mathrm{d}P=\tilde{\mb{I}}_w^{-1}\tilde{W}$,
	where $_\mathrm{d}P$ is defined as in Section~\ref{preliminaries}. Note also that is
	$_\mathrm{d}P$ is sub-stochastic, since the rows referring to nodes pointing to the destination sum to less than one.
	The inverse of $Q$ may be written as follows.
	\begin{equation*}
		\begin{aligned}
		Q^{-1}&=(\tilde{\mb{I}}_w-\tilde{W})^{-1}=(\tilde{\mb{I}}_w(\mb{I}-{_\mathrm{d}P}))^{-1}=(\mb{I}-{_\mathrm{d}P})^{-1}\tilde{\mb{I}}_w^{-1}\\
		&=\sum_{t=0}^\infty ({_\mathrm{d}P})^t \tilde{\mb{I}}_w^{-1}=  {_\mathrm{d}G}\tilde{\mb{I}}_w^{-1},
		\end{aligned}
	\end{equation*}
	where the first equivalence follows from $Q={_\dd L}$, and the penultimate one follows from connectedness of $\mathcal{G}_R$ and \eqref{G}. We now construct $\hat{Q}^{-1} \in \mathds{R}^{\mathrm{N} \times \mathrm{N}}$ and ${_\mathrm{d}\hat{G}} \in \mathds{R}^{\mathrm{N} \times \mathrm{N}}$ by adding a zero column and a zero row to $Q^{-1}$ and ${_\mathrm{d}G}$, and construct $\hat{K} \in \mathds{R}^{\mathrm{E} \times \mathrm{N}}$ by adding a zero column to $K$ corresponding to the destination. By construction,  $\hat{Q}^{-1}={_\mathrm{d}\hat{G}}\tilde{\mb{I}}_w^{-1}$. Consider now a link $e$ with $\xi(e)=i$, $\theta(e)=j$. It follows
	\begin{equation}
		\label{eq:rij}
		\begin{aligned}
			(B_-^e)^T Q^{-1} B_-^e &= (B^e)^T \hat{Q}^{-1} B^e \\ &=\left(\delta^{(i)}-\delta^{(j)}\right)^T {_\mathrm{d}\hat{G}\tilde{\mb{I}}_w^{-1}} (\delta^{(i)} - \delta^{(j)})\\
			&=\frac{_\mathrm{d}\hat{G}_{ii}-{_\mathrm{d}\hat{G}}_{ji}}{w_{i}}+\frac{_\mathrm{d}\hat{G}_{jj}-{_\mathrm{d}\hat{G}}_{ij}}{w_{j}}=r_{e},
		\end{aligned}
	\end{equation}
	where we recall that $r_e$ denotes the effective resistance of link $M(e)=\{i,j\}$ in $\mc{L}$, and the last equivalence follows from \eqref{reff} and from noticing that the definition of $_\mathrm{d}{\hat{G}}$ is coherent with \eqref{green0}. Let $v_-$ denote the restriction of $v$ on $\mc{N} \setminus \{\dd\}$. Definition \ref{def:voltage} and $Q={_\dd L}$ imply that 
	\begin{equation}
		\label{eq:compact_v}
		v_- = m Q^{-1}\delta^{(\mathrm{o})}.
	\end{equation}
	Plugging this equivalence and \eqref{eq:rij} in \eqref{eq:delta_gamma}, we get
	\begin{equation}
		\label{eq:delta_gamma2}
		\begin{aligned}
				\gamma^*_\mathrm{o} - \gamma^*_\mathrm{o}(\control_e \delta^{(e)})
			&= -\frac{\control_e}{a_e}(\delta^{(\mathrm{o})})^T Q^{-1}B_-^e(\delta^{(e)})^T b +\\
			&+ \frac{(\delta^{(\mathrm{o})})^TQ^{-1}B_-^e (B_{-}^e)^T Q^{-1}}{\frac{a_e}{\control_e}+(B_{-}^e)^T Q^{-1} B_{-}^e}\cdot\\
			&\cdot \left(K^T b + \frac{\control_e}{a_e}B_-^e(\delta^{(e)})^T b + \nu_- \right)\\
			&=-\frac{\control_e}{m} \frac{b_e}{a_e}(v_i-v_j)+\\
			&+\frac{1}{m}\frac{v_i-v_j}{\frac{a_e}{\control_e}+r_e} \big((B_-^e)^T \gamma^*_-+\control_e\frac{b_e}{a_e} r_e\big)\\
			&= \frac{1}{m}\frac{v_i-v_j}{\frac{a_e}{\control_e}+r_e} \left(-b_e+\gamma_i^*-\gamma_j^*\right)\\
			& = \frac{1}{m}\frac{v_i-v_j}{\frac{1}{\control_e}+\frac{r_e}{a_e}} f_e^*
		\end{aligned}
	\end{equation}
	where the second equivalence follows from KKT conditions $Q^{-1}(K^Tb+\nu_-)=\gamma^*_-$, the last one from $\gamma_i^*-\gamma_j^*=a_ef_e^*+b_e$, and $v$ is used instead of $v_-$, coherently with the convention $\delta^{(\mathrm{d})}=0\cdot \mb{1}$. The statement then follows from Lemma~\ref{cost} from $\gamma_{\dd}=0$, and from Ohm's law, i.e., $v_i-v_j = a_ey_e$.

\subsection{Proof of Proposition~\ref{sp_prp}}
A sufficient condition under which $\mathcal{E}_+=\emptyset$ is that the first $\mathrm{E}$ components of \eqref{linear}, corresponding to equilibrium link flows, are nonnegative. Indeed,
since \eqref{convex_prob} is strictly convex, if the flow $f^*$ obtained by \eqref{linear} is non-negative, then $f^*$ is feasible and is the unique Wardrop equilibrium, with $\mathbf{\lambda^*}=\mathbf{0}$. Links $e$ with $\lambda_e^*>0$ are those such that $f_e^*$ computed by \eqref{linear} is strictly negative.
Hence, we aim at finding conditions under which $f_e^*\ge 0$ for every $e$ in $\mc E$ according to \eqref{linear}. Let us define $\tilde{v}=v/m$.
From \eqref{linear}, \eqref{eq:compact_v}, and $\nu_-=m \delta^{(\mathrm{o})}$, it follows that for every link $e$,
\begin{align*}
	f_e^*&=-\frac{b_e}{a_e}+[KQ^{-1}K^T]_{e:}b+[KQ^{-1}]_{e:}(\nu_-)\\
	&=-\frac{b_e}{a_e}+[KQ^{-1}K^T]_{e:}b+m\frac{\tv_{\xi(e)}-\tv_{\theta(e)}}{a_e}.
\end{align*}
Let $\overline{m}_e = (b_e-a_e[K Q^{-1}K^T]_{e:}b)/ \Delta \tv_e$.
If $\Delta \tv_e>0$, then for every $m\ge\overline{m}_e$
it holds $f_e^*\ge0$, which in turn implies that if $m \ge \overline{m}:= \{\overline{m}_e\}_{e=1}^{\mathrm{E}}$, then $\mathcal{E}_+=\emptyset$.
Moreover, if the delays are linear, $\Delta \tv_e \ge 0$ implies $f_e^*\ge 0$ and $\mathcal{E}_+=\emptyset$ for every $m \ge 0$, because $b=\mb{0}$. We have now to prove that $\Delta \tv_e>0$. Note by Ohm's law that $\Delta \tv_e \cdot a_e = \ty_e$, where $\ty_e$ denotes the current flowing on $\mc G_R$ from node $\xi(e)$ to node $\theta(e)$ when unitary current is injected from $\oo$ to $\dd$. Then, it suffices to show that $\ty_e>0$. To this end, observe that if the transportation network is series-parallel, it has single link $e: \xi(e)=\mathrm{o}, \theta(e)=\mathrm{d}$, or it can obtained by connecting in series or in parallel two series-parallel networks. Thus, a series-parallel network can be reduced to a single link from $\oo$ to $\dd$ by recursively i) merging two links $e_1$ and $e_2$ connected in series (i.e., $\xi(e_2)=\theta(e_1)$) into a single link $e_3$, or ii) merging two links $e_1$ and $e_2$ connected in parallel, i.e., with same head and tail, into a single link $e_3$. The transformation (i) results in an associated resistor network where the links $M(e_1)$ and $M(e_2)$ are replaced by their series composition $M(e_3)=\{\xi(e_1),\theta(e_2)\}$ with current $\ty_{e_3}=\ty_{e_1}=\ty_{e_2}$. Instead, the transformation (ii) results in an associated resistor network where the links $M(e_1)$ and $M(e_2)$ are replaced by their parallel composition $M(e_3)$, with $\ty_{e_3}>0$ if and only if $\ty_{e_1},\ty_{e_2}>0$. Thus, in both the cases (i) and (ii), $\ty_{e_3}>0$ if and only if $\ty_{e_1}>0$ and $\ty_{e_2}>0$. Obviously, when the transportation network is reduced to a single link from $\mathrm{o}$ to $\mathrm{d}$, the flow on the unique link is positive because $m>0$. Then, by applying those arguments recursively, for every link $e$ in $\mathcal{E}$, we get $\ty_{e}>0$, which implies by Ohm's law that $\Delta \tv_e>0$. Thus, if $m \ge \overline{m}$ then $f_e^* \ge 0$ and $\mathcal{E}_+ =\emptyset$, concluding the proof.

\subsection{Proof of Theorem~\ref{thm2}}
Consider an intervention $\control = \control_e \delta^{(e)}$. Then,
\begin{equation*}
	\begin{aligned}
		|\Delta C^{(\control)}-\Delta C^{(\control)}_d| &= a_ef_e^*\Bigg|\frac{y_e}{\frac{1}{\control_e}+\frac{r_e}{a_e}}-\frac{y_e} {\frac{1}{\control_e}+\frac{r_e^{U_d}+r_e^{L_d}}{2a_e}} \Bigg|\\
		& =\Bigg|\frac{a_ef_e^*y_e}{\frac{1}{\control_e}+\frac{r_e}{a_e}}\Bigg|\cdot \Bigg|\frac{\frac{r_e^{U_d}+r_e^{L_d}-2r_e}{2a_e}}{\frac{1}{\control_e}+\frac{r_e^{U_d}+r_{e}^{L_d}}{2a_e}}\Bigg|,
	\end{aligned}
\end{equation*}
Notice also that 
\begin{equation*}
	\begin{aligned}
	&\frac{|r^{U_d}_{e}+r^{L_d}_{e}-2r_{e}|}{a_e} \le \frac{|r^{U_d}_{e}-r_{e}|+|r_{e}-r^{L_d}_{e}|}{a_e}\\
	=&\frac{r^{U_d}_{e}-r_{e}+r_{e}-r^{L_d}_{e}}{a_e}=\frac{r^{U_d}_{e}-r^{L_d}_{e}}{a_e}=\epsilon_{ed}.
	\end{aligned}
\end{equation*}
Putting those two together, and using \eqref{el_form}, we get
\begin{equation*}
	\begin{aligned}
	\bigg|\frac{\Delta C^{(\control)}-\Delta C_d^{(\control)}}{\Delta C^{(\control)}}\bigg| &\le \frac{\epsilon_{ed}}{2\Big(\frac{1}{\control_e}+\frac{r_{e}^{U_d}+r_{e}^{L_d}}{2a_e}\Big)} \\
	&\le \frac{\epsilon_{ed}}{2\Big(\frac{1}{\control_e}+\frac{1}{w^* \cdot a_e}\Big)},
	\end{aligned}
\end{equation*}
where the last inequality follows from \eqref{remark}. Finally, \eqref{min_gain} follows from Theorem~\ref{thm} and $r_{e}^{U_d}\ge r_{e}^{L_d}$, concluding the proof.

\subsection{Proof of Proposition~\ref{speed}}
The cut and shorted networks are obtained by finding the neighbors within distance $d$ and $d+1$ from ${i,j}$, respectively. The neighbors of a node $i$ can be found by checking the non-zero elements of $W(i,:)$. The neighbors within distance $d$ can be found by iterating such operation $d$ times. Hence, the time to construct the cut and the shorted network depends on the local structure, which, under Assumption~\ref{sparse}, does not depend on the network size. Since the bounds of the effective resistance are computed on these subnetwork, their time complexity and tightness depends on local structure, which, under Assumption~\ref{sparse}, is independent of the network size.

\subsection{Proof of Proposition~\ref{gap}}
We introduce the following notation:
\begin{itemize}
	\item The index $U_d$ and $L_d$ indicate that the random walk takes place over $\mathcal{G}^{U_d}_{l}$ and $\mathcal{G}^{L_d}_{l}$, respectively. So, for instance, $_kG_{ij}^{U_d}$ denotes the expected number of times that the random walk on the network $\mathcal{G}^{U_d}_{l}$, starting from $i$, hits $j$ before hitting $k$.
	
	\item $p_i(T_j = T_\mc{S})$, with $j$ in $\mc{S}$, denotes the probability that the random walk starting from $i$ hits the node $j$ in $\mc{S}$ before hitting any other node in $\mc{S}$.
\end{itemize} 
By applying \eqref{reff} to the effective resistance of link $l=\{i,j\}$ in the shorted and the cut network, it follows
\begin{equation*}
	\begin{gathered}
		r^{U_d}_{l}= \frac{_jG^{U_{d}}_{ii}}{w_{i}}, \quad \quad
		r^{L_d}_{l}= \frac{_jG^{L_{d}}_{ii}}{w_{i}},
	\end{gathered}
\end{equation*}
where we recall that $_jG^{U_{d}}_{ii}$ and $_jG^{L_{d}}_{ii}$ are the expected number of visits on $i$, before hitting $j$, starting from $i$, of the random walk defined on $\mathcal{G}^{U_d}_{l}$ and $\mathcal{G}^{L_d}_{l}$ respectively. The visits on $i$ before hitting $j$ can be divided in two disjoint sets: the visits before hitting $j$ and before visiting any node in $\mc{N}_d$, and the visits before hitting $j$ but after at least a node in $\mc{N}_d$ has been visited.
Let $G^{<\mc{N}_d}_{ii}$ denote the expected number of visits to $i$, starting from $i$, before hitting any node in $\mc{N}_d$ and before hitting the absorbing node $j$ (for simplicity of notation we omit the index $j$ from now on). Note that $\mathcal{G}^{U_d}_{l}$ and $\mathcal{G}^{L_d}_{l}$ differ only in the node $s$, which is the node obtained by shorting all the nodes at distance greater than $d$ from $i$ and $j$. Since $s$ cannot be reached before hitting nodes in $\mc{N}_d$ before, $G^{<\mc{N}_d}_{ii}$ is equivalent when computed on $\mathcal{G}^{U_d}_{l}$ and $\mathcal{G}^{L_d}_{l}$. Thus, we can write the following decomposition,
\begin{equation*}
	\begin{aligned}
		&G_{ii}^{U_d}=G_{ii}^{<\mc{N}_d}+G_{ii}^{U>\mc{N}_d}, \\
		&G_{ii}^{L_d}=G_{ii}^{<\mc{N}_d}+G_{ii}^{L>\mc{N}_d},
	\end{aligned}
\end{equation*}
where $G_{ii}^{U>\mc{N}_d}$ and $G_{ii}^{L>\mc{N}_d}$ indicate respectively the expected visits in $i$, starting from $i$, before hitting $j$ and after hitting any node in $\mc{N}_d$, on $\mc{G}_l^{U_d}$ and $\mc{G}_l^{L_d}$ respectively.
This implies by \eqref{reff}
\begin{equation}
	r_{l}^{U_d}-r_{l}^{L_d}=\frac{G_{ii}^{U>\mc{N}_d}-G_{ii}^{L>\mc{N}_d}}{w_{i}}.
	\label{tight_G}
\end{equation}
Notice that $G_{ii}^{U>\mc{N}_d}$ can be written as the sum over the nodes $g$ in $\mc{N}_d$ of the probability, starting from $i$, of hitting $g$ and going back to $i$ without hitting $j$, multiplied by the expected number of visits on $i$ starting from $i$, before hitting $j$, which is the derivative of a geometric sum. 
Therefore,
\begin{equation*}
	\begin{split}
		G_{ii}^{U>\mc{N}_d} &=\sum_{g \in \mc{N}_d} \underbrace{p_i(T_g = T_{j \cup \mc{N}_d})}_{\text{(1)}} \underbrace{p^{U_d}_g(T_i < T_j)}_{\text{(2)}}\cdot \\
		&\cdot \sum_{k=1}^\infty k \underbrace{\big(p^{U_d}_i(T_i^+ < T_j)\big)^{k-1}}_{\text{(3)}}\underbrace{\big(1-p^{U_d}_i(T_i^+ < T_j)\big)}_{\text{(4)}} \\
		&=\frac{\sum_{g \in \mc{N}_d}  p_i(T_g = T_{j \cup \mc{N}_d}) p^{U_d}_g(T_i < T_j)}{1-p^{U_d}_i(T_i^+ < T_j)},
	\end{split}
\end{equation*}
where:
\begin{enumerate}
	\item probability of hitting $g$ before hitting $j$ and any other node in $\mc{N}_d$ starting from $i$;
	\item probability of hitting $i$ before $j$ starting from $g$;
	\item probability of hitting $k-1$ times $i$ before hitting $j$ starting from $i$;
	\item probability of hitting $j$ before returning in $i$ starting from $i$.
\end{enumerate}
Similarly,
\begin{equation*}
	\begin{split}
		G_{ii}^{L>\mc{N}_d}&=\sum_{g \in \mc{N}_d} p_i(T_g = T_{j \cup \mc{N}_d}) p^{L_d}_g(T_i < T_j)\cdot \\
		&\cdot \sum_{k=1}^\infty k \big(p^{L_d}_i( T_i^+ < T_j)\big)^{k-1}\big(1-p^{L_d}_i( T_i^+ < T_j)\big) \\
		&=\frac{\sum_{g \in \mc{N}_d}  p_i(T_g = T_{j \cup \mc{N}_d}) p^{L_d}_g(T_i < T_j)}{1-p^{L_d}_i( T_i^+ < T_j)}.
	\end{split}
\end{equation*}
Substituting in \eqref{tight_G}, we get
\begin{equation*}
	\begin{aligned}
	r^{U_d}_{l}-r^{L_d}_{l}&=\frac{1}{w_{i}}\sum_{g \in \mc{N}_d} p_i(T_g = T_{j \cup \mc{N}_d})\cdot \\ &\cdot \bigg(\frac{p^{U_d}_g( T_i < T_j)}{1-p^{U_d}_i( T_i^+ < T_j)}
	-\frac{p^{L_d}_g(T_i < T_j)}{1-p^{L_d}_i( T_i^+ < T_j)}\bigg).
	\end{aligned}
\end{equation*}
From \eqref{reff}, it follows
\begin{align*}
	&r^{U_d}_{l}=\frac{1}{w_{i}p^{U_d}_i( T_j < T_i^+)}=\frac{1}{w_{i}\big(1-p^{U_d}_i( T_i^+ < T_j)\big)}, \\
	&r^{L_d}_{l}=\frac{1}{w_{i}p^{L_d}_i(T_j < T_i^+)}=\frac{1}{w_{i}\big(1-p^{L_d}_i( T_i^+ < T_j)\big)}.
\end{align*}
Therefore, $r^{U_d}_{l}-r^{L_d}_{l}$ reads
\begin{equation*}
	\begin{split}
		&\sum_{g \in \mc{N}_d} p_i(T_g = T_{j \cup \mc{N}_d}) \big(p^{U_d}_g(T_i < T_j)r_{l}^{U_d}
		-p^{L_d}_g(T_i < T_j)r_{l}^{L_d}\big)\\
		=&\sum_{g \in \mc{N}_d} p_i(T_g = T_{j \cup \mc{N}_d})\big(p^{U_d}_g(T_i < T_j)-p^{L_d}_g(T_i < T_j)\big)r_{l}^{U_d}+\\
		+&\sum_{g \in \mc{N}_d} p_i(T_g = T_{j \cup \mc{N}_d})
		p^{L_d}_g(T_i < T_j)(r_{l}^{U_d}-r_{l}^{L_d}) \\
		\le&\sum_{g \in \mc{N}_d} p_i(T_g = T_{j \cup \mc{N}_d})\big(p^{U_d}_g(T_i < T_j)-p^{L_d}_g(T_i < T_j)\big)r_{l}^{U_d}+\\
		+&\sum_{g \in \mc{N}_d} p_i(T_g = T_{j \cup \mc{N}_d})
		(r_{l}^{U_d}-r_{l}^{L_d})\\
		=&\sum_{g \in \mc{N}_d} p_i(T_g = T_{j \cup \mc{N}_d})\big(p^{U_d}_g(T_i < T_j)-p^{L_d}_g(T_i < T_j)\big)r_{l}^{U_d}+\\
		+ & p_i(T_{\mc{N}_d} < T_j)
		(r_{l}^{U_d}-r_{l}^{L_d}),
	\end{split}
\end{equation*}
where the last inequality follows from $p^L_g(T_i < T_j) \le 1$ and the last equality from the fact that $p_i( T_{\mc{N}_d} < T_j)=\sum_{g \in \mc{N}_d} p_i(T_g = T_{j \cup \mc{N}_d})$.
It thus follows
\begin{equation*}
	\begin{split}
		r^{U_d}_{l}-r^{L_d}_{l}&\scriptstyle{\le \frac{\sum_{g \in \mc{N}_d} p_i(T_g = T_{j \cup \mc{N}_d})\left(p^U_g(T_i < T_j)-p^L_g(T_i < T_j)\right)r_{l}^{U_d}}{1-p_i(T_{\mc{N}_d} < T_j)}} \\
		&\scriptstyle{\le \sum_{g \in \mc{N}_d} p_i(T_g = T_{j \cup \mc{N}_d})\big(p^U_g(T_i < T_j)-p^L_g(T_i < T_j)\big)r_{l}^{U_d}\frac{w_{i}}{W_{ij}}}\\
		&\scriptstyle{\le p_i(T_{\mc{N}_d} < T_j) \underset{g \in \mc{N}_d} \max \big(p^U_g(T_i < T_j)-p^L_g(T_i < T_j)\big)  \frac{w_{i}}{(W_{ij})^2}.}
	\end{split}
\end{equation*}
where the second inequality follows from $1-p_i(T_{\mc{N}_d} < T_j)=p_i(T_j < T_{\mc{N}_d})\ge P_{ij}=W_{ij}/w_{i}$, and the last one from $r_{ij}^{U_d}\le 1/W_{ij}$ (as shown in \eqref{remark}) and from $p_i( T_{\mc{N}_d} < T_j)=\sum_{g \in \mc{N}_d} p_i(T_g = T_{j \cup \mc{N}_d})$.

\subsection{More details on Section~\ref{section_tree}}
\label{app:tree}
\begin{figure}
	\centering
	\begin{tikzpicture}[scale=0.75, transform shape]
			\node[draw, circle, fill=green] (1) at (0,0) {$i$};
			\node[draw, circle, fill=green] (2) at (1.5 ,0) {$j$};
			\node[draw, circle] (3) at (-1.5,0) {};
			\node[draw, circle] (4) at (3,0) {};
			\node[draw, circle] (5) at (-3,0) {};
			\node[draw, circle] (6) at (4.5,0) {};
			\node[draw=none] (7) at (-4.5,0) {};
			\node[draw=none] (8) at (6,0) {};
			\path [->, >=latex]  (1) edge[bend left =20]   node [above] {$1/3$} (2);
			\path [->, >=latex]  (2) edge[bend left =20]   node [below] {$1/3$} (1);
			\path [->, >=latex]  (2) edge[bend left =20]   node [above] {$2/3$} (4);
			\path [->, >=latex]  (1) edge[bend left =20]   node [below] {$2/3$} (3);
			\path [->, >=latex]  (4) edge[bend left =20]   node [above] {$2/3$} (6);
			\path [->, >=latex]  (3) edge[bend left =20]   node [below] {$2/3$} (5);
			\path [->, >=latex, dashed]  (6) edge[bend left =20]   node [above] {$2/3$} (8);
			\path [->, >=latex, dashed]  (5) edge[bend left =20]   node [below] {$2/3$} (7);
			\path [->, >=latex]  (4) edge[bend left =20]   node [below] {$1/3$} (2);
			\path [->, >=latex]  (3) edge[bend left =20]   node [above] {$1/3$} (1);
			\path [->, >=latex]  (6) edge[bend left =20]   node [below] {$1/3$} (4);
			\path [->, >=latex]  (5) edge[bend left =20]   node [above] {$1/3$} (3);
			\path [->, >=latex, dashed]  (8) edge[bend left =20]   node [below] {$1/3$} (6);
			\path [->, >=latex, dashed]  (7) edge[bend left =20]   node [above] {$1/3$} (5);
	\end{tikzpicture}
	\caption{The double tree network is equivalent to a biased random walk like this. \label{biased_rw}}
\end{figure}
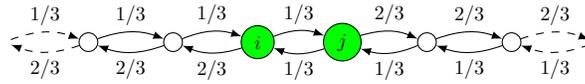
\begin{figure}
		\centering{
		(a) \hspace{0.35cm}
		\begin{tikzpicture}[scale=0.5, transform shape]
			\node[draw, circle, fill=green] (1) at (-2,0)  {$i$};
			\node[draw, circle, fill=green] (2) at (2,0) {$j$};
			\node[draw, circle] (3) at (-2.8,1)  {};
			\node[draw, circle] (4) at (-1.2,1) {};
			\node[draw, circle] (5) at (1.2,1) {};
			\node[draw, circle] (6) at (2.8,1) {};
			\node[draw, circle] (7) at (-3.2,2)  {};
			\node[draw, circle] (8) at (-2.4,2) {};
			\node[draw, circle] (9) at (-1.6,2) {};
			\node[draw, circle] (10) at (-0.8,2) {};
			\node[draw, circle] (11) at (0.8,2)  {};
			\node[draw, circle] (12) at (1.6,2) {};
			\node[draw, circle] (13) at (2.4,2) {};
			\node[draw, circle] (14) at (3.2,2) {};
			
			\node[draw=none] (15) at (-3.5,3)  {};
			\node[draw=none] (16) at (-2.9,3) {};
			\node[draw=none] (17) at (-2.7,3) {};
			\node[draw=none] (18) at (-2.1,3) {};
			\node[draw=none] (19) at (-1.9,3)  {};
			\node[draw=none] (20) at (-1.3,3) {};
			\node[draw=none] (21) at (-1.1,3) {};
			\node[draw=none] (22) at (-0.5,3) {};
			\node[draw=none] (23) at (0.5,3)  {};
			\node[draw=none] (24) at (1.1,3) {};
			\node[draw=none] (25) at (1.3,3) {};
			\node[draw=none] (26) at (1.9,3) {};
			\node[draw=none] (27) at (2.1,3)  {};
			\node[draw=none] (28) at (2.7,3) {};
			\node[draw=none] (29) at (2.9,3) {};
			\node[draw=none] (30) at (3.5,3) {};
			
			\path   (1) edge [bend right=0]
			node [above] {} (2);
			
			\path   (1) edge [bend right=0]
			node [above] {} (3);
			\path   (1) edge [bend left=0]
			node [above] {} (4);
			\path (2) edge [bend right=0]
			node [above] {} (5);
			\path   (2) edge [bend right=0]
			node [below] {} (6);
			
			\path   (3) edge [bend left=0]
			node [below] {} (7);
			\path  (3) edge [bend left=0]
			node [below] {} (8);
			\path   (4) edge [bend right=0]
			node [above] {} (9);
			\path   (4) edge [bend right=0]
			node [above] {} (10);
			\path   (5) edge [bend left=0]
			node [above] {} (11);
			\path (5) edge [bend right=0]
			node [above] {} (12);
			\path   (6) edge [bend right=0]
			node [below] {} (13);
			\path   (6) edge [bend left=0]
			node [below] {} (14);
			
			\path   (7) edge [dashed]
			node [above] {} (15);
			\path   (7) edge [dashed]
			node [above] {} (16);
			\path   (8) edge [dashed]
			node [above] {} (17);
			\path (8) edge [dashed]
			node [above] {} (18);
			\path   (9) edge [dashed]
			node [below] {} (19);
			\path   (9) edge [dashed]
			node [below] {} (20);
			\path  (10) edge [dashed]
			node [below] {} (21);
			\path   (10) edge [dashed]
			node [above] {} (22);
			\path   (11) edge [dashed]
			node [above] {} (23);
			\path   (11) edge [dashed]
			node [above] {} (24);
			\path (12) edge [dashed]
			node [above] {} (25);
			\path   (12) edge [dashed]
			node [below] {} (26);
			\path   (13) edge [dashed]
			node [below] {} (27);
			\path (13) edge [dashed]
			node [above] {} (28);
			\path   (14) edge [dashed]
			node [below] {} (29);
			\path   (14) edge [dashed]
			node [below] {} (30);
		\end{tikzpicture}
	\qquad}

		\centering{
	(b) \qquad
		\begin{tikzpicture}[scale=0.5, transform shape]
			\node[draw, circle, fill=green] (1) at (-2,0)  {$i$};
			\node[draw, circle, fill=green] (2) at (2,0) {$j$};
			\node[draw, circle] (3) at (-2.8,1)  {};
			\node[draw, circle] (4) at (-1.2,1) {};
			\node[draw, circle] (5) at (1.2,1) {};
			\node[draw, circle] (6) at (2.8,1) {};
			\node[draw, circle, fill=red] (7) at (-3.2,2)  {};
			\node[draw, circle, fill=red] (8) at (-2.4,2) {};
			\node[draw, circle, fill=red] (9) at (-1.6,2) {};
			\node[draw, circle, fill=red] (10) at (-0.8,2) {};
			\node[draw, circle, fill=red] (11) at (0.8,2)  {};
			\node[draw, circle, fill=red] (12) at (1.6,2) {};
			\node[draw, circle, fill=red] (13) at (2.4,2) {};
			\node[draw, circle, fill=red] (14) at (3.2,2) {};
			
			\path   (1) edge [bend right=0]
			node [above] {} (2);
			
			\path   (1) edge [bend right=0]
			node [above] {} (3);
			\path   (1) edge [bend left=0]
			node [above] {} (4);
			\path (2) edge [bend right=0]
			node [above] {} (5);
			\path   (2) edge [bend right=0]
			node [below] {} (6);
			
			\path   (3) edge [bend left=0]
			node [below] {} (7);
			\path  (3) edge [bend left=0]
			node [below] {} (8);
			\path   (4) edge [bend right=0]
			node [above] {} (9);
			\path   (4) edge [bend right=0]
			node [above] {} (10);
			\path   (5) edge [bend left=0]
			node [above] {} (11);
			\path (5) edge [bend right=0]
			node [above] {} (12);
			\path   (6) edge [bend right=0]
			node [below] {} (13);
			\path   (6) edge [bend left=0]
			node [below] {} (14);
		\end{tikzpicture}
	\qquad}

\centering{
(c) \qquad
		\begin{tikzpicture}[scale=0.5, transform shape]
			\node[draw, circle, fill=green] (1) at (-2,0)  {$i$};
			\node[draw, circle, fill=green] (2) at (2,0) {$j$};
			\node[draw, circle] (3) at (-2.8,1)  {};
			\node[draw, circle] (4) at (-1.2,1) {};
			\node[draw, circle] (5) at (1.2,1) {};
			\node[draw, circle] (6) at (2.8,1) {};
			\node[draw, circle, fill=red] (7) at (-3.2,2)  {};
			\node[draw, circle, fill=red] (8) at (-2.4,2) {};
			\node[draw, circle, fill=red] (9) at (-1.6,2) {};
			\node[draw, circle, fill=red] (10) at (-0.8,2) {};
			\node[draw, circle, fill=red] (11) at (0.8,2)  {};
			\node[draw, circle, fill=red] (12) at (1.6,2) {};
			\node[draw, circle, fill=red] (13) at (2.4,2) {};
			\node[draw, circle, fill=red] (14) at (3.2,2) {};
			
			\node[draw, circle] (15) at (0,3.6)  {s};
			
			\path   (1) edge [bend right=0]
			node [above] {} (2);
			
			\path   (1) edge [bend right=0]
			node [above] {} (3);
			\path   (1) edge [bend left=0]
			node [above] {} (4);
			\path (2) edge [bend right=0]
			node [above] {} (5);
			\path   (2) edge [bend right=0]
			node [below] {} (6);
			
			\path   (3) edge [bend left=0]
			node [below] {} (7);
			\path  (3) edge [bend left=0]
			node [below] {} (8);
			\path   (4) edge [bend right=0]
			node [above] {} (9);
			\path   (4) edge [bend right=0]
			node [above] {} (10);
			\path   (5) edge [bend left=0]
			node [above] {} (11);
			\path (5) edge [bend right=0]
			node [above] {} (12);
			\path   (6) edge [bend right=0]
			node [below] {} (13);
			\path   (6) edge [bend left=0]
			node [below] {} (14);
			
			\path   (7) edge [bend left=5]
			node [below] {} (15);
			\path  (7) edge [bend left=-5]
			node [below] {} (15);
			\path   (8) edge [bend right=5]
			node [above] {} (15);
			\path   (8) edge [bend right=-5]
			node [above] {} (15);
			\path   (9) edge [bend right=5]
			node [above] {} (15);
			\path (9) edge [bend right=-5]
			node [above] {} (15);
			\path   (10) edge [bend left=5]
			node [below] {} (15);
			\path   (10) edge [bend left=-5]
			node [below] {} (15);
			\path   (14) edge [bend left=5]
			node [below] {} (15);
			\path  (14) edge [bend left=-5]
			node [below] {} (15);
			\path   (13) edge [bend right=5]
			node [above] {} (15);
			\path   (13) edge [bend right=-5]
			node [above] {} (15);
			\path   (12) edge [bend right=5]
			node [above] {} (15);
			\path (12) edge [bend right=-5]
			node [above] {} (15);
			\path   (11) edge [bend left=5]
			node [below] {} (15);
			\path   (11) edge [bend left=-5]
			node [below] {} (15);
		\end{tikzpicture}
	\qquad}

\centering{
(d) \qquad
		\begin{tikzpicture}[scale=0.5, transform shape]
			\node[draw, circle, fill=green] (1) at (-2,0)  {$i$};
			\node[draw, circle, fill=green] (2) at (2,0) {$j$};
			\node[draw, circle] (3) at (-2.8,1)  {};
			\node[draw, circle] (4) at (-1.2,1) {};
			\node[draw, circle] (5) at (1.2,1) {};
			\node[draw, circle] (6) at (2.8,1) {};
			\node[draw, circle, fill=red] (7) at (-3.2,2)  {};
			\node[draw, circle, fill=red] (8) at (-2.4,2) {};
			\node[draw, circle, fill=red] (9) at (-1.6,2) {};
			\node[draw, circle, fill=red] (10) at (-0.8,2) {};
			\node[draw, circle, fill=red] (11) at (0.8,2)  {};
			\node[draw, circle, fill=red] (12) at (1.6,2) {};
			\node[draw, circle, fill=red] (13) at (2.4,2) {};
			\node[draw, circle, fill=red] (14) at (3.2,2) {};
			
			\node[draw, circle, fill=yellow] (15) at (0,2)  {};
			\node[draw, circle, fill=yellow] (16) at (0,2.8) {};
			\node[draw, circle, fill=yellow] (17) at (0,3.6) {};
			\node[draw, circle, fill=yellow] (18) at (0,4.4) {};
			
			\path   (1) edge [bend right=0]
			node [above] {} (2);
			
			\path   (1) edge [bend right=0]
			node [above] {} (3);
			\path   (1) edge [bend left=0]
			node [above] {} (4);
			\path (2) edge [bend right=0]
			node [above] {} (5);
			\path   (2) edge [bend right=0]
			node [below] {} (6);
			
			\path   (3) edge [bend left=0]
			node [below] {} (7);
			\path  (3) edge [bend left=0]
			node [below] {} (8);
			\path   (4) edge [bend right=0]
			node [above] {} (9);
			\path   (4) edge [bend right=0]
			node [above] {} (10);
			\path   (5) edge [bend left=0]
			node [above] {} (11);
			\path (5) edge [bend right=0]
			node [above] {} (12);
			\path   (6) edge [bend right=0]
			node [below] {} (13);
			\path   (6) edge [bend left=0]
			node [below] {} (14);
			
			\path   (7) edge [bend left=10]
			node [below] {} (18);
			\path  (7) edge [bend left=-10]
			node [below] {} (18);
			\path   (8) edge [bend right=10]
			node [above] {} (17);
			\path   (8) edge [bend right=-10]
			node [above] {} (17);
			\path   (9) edge [bend right=10]
			node [above] {} (16);
			\path (9) edge [bend right=-10]
			node [above] {} (16);
			\path   (10) edge [bend left=10]
			node [below] {} (15);
			\path   (10) edge [bend left=-10]
			node [below] {} (15);
			\path   (14) edge [bend left=10]
			node [below] {} (18);
			\path  (14) edge [bend left=-10]
			node [below] {} (18);
			\path   (13) edge [bend right=10]
			node [above] {} (17);
			\path   (13) edge [bend right=-10]
			node [above] {} (17);
			\path   (12) edge [bend right=10]
			node [above] {} (16);
			\path (12) edge [bend right=-10]
			node [above] {} (16);
			\path   (11) edge [bend left=10]
			node [below] {} (15);
			\path   (11) edge [bend left=-10]
			node [below] {} (15);
		\end{tikzpicture}
	\qquad}

	\caption{From above to below: $(a)$ the double tree network; $(b)$ the network cut at distance $2$ from $l=\{i,j\}$; $(c)$ the network shorted at distance $2$ from $l=\{i,j\}$; $(d)$ a network equivalent to the shorted one. In red, the nodes at distance $2$.}
	\label{counterex3}
\end{figure}
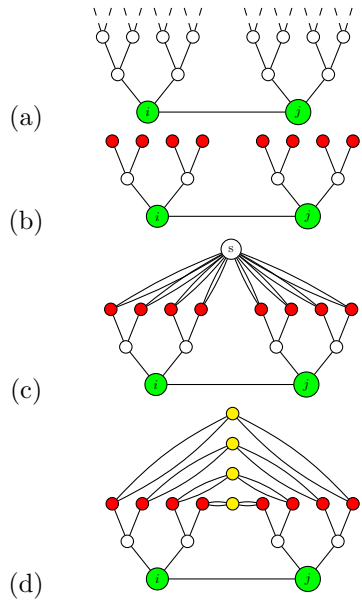
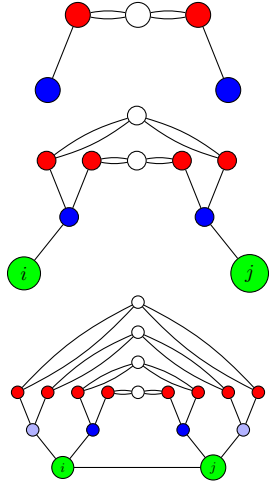
\begin{figure}[h!]
	\centering
		\begin{tikzpicture}
			\node[draw, circle, fill=blue] (4) at (-1.2,1) {};
			\node[draw, circle, fill=blue] (5) at (1.2,1) {};
			\node[draw, circle, fill=red] (10) at (-0.8,2) {};
			\node[draw, circle, fill=red] (11) at (0.8,2)  {};
			\node[draw, circle] (15) at (0,2)  {};

			\path   (4) edge [bend right=0]
			node [below] {} (10);
			\path   (5) edge [bend left=0]
			node [below] {} (11);
			
			\path   (10) edge [bend left=10]
			node [below] {} (15);
			\path   (10) edge [bend left=-10]
			node [below] {} (15);
			\path   (11) edge [bend left=10]
			node [below] {} (15);
			\path   (11) edge [bend left=-10]
			node [below] {} (15);
		\end{tikzpicture}\\
	\centering
		\begin{tikzpicture}[scale=0.75, transform shape]
			\node[draw, circle, fill=green] (1) at (-2,0)  {$i$};
			\node[draw, circle, fill=green] (2) at (2,0) {$j$};
			\node[draw, circle, fill=blue] (4) at (-1.2,1) {};
			\node[draw, circle, fill=blue] (5) at (1.2,1) {};
			\node[draw, circle, fill=red] (9) at (-1.6,2) {};
			\node[draw, circle, fill=red] (10) at (-0.8,2) {};
			\node[draw, circle, fill=red] (11) at (0.8,2)  {};
			\node[draw, circle, fill=red] (12) at (1.6,2) {};
			\node[draw, circle] (15) at (0,2)  {};
			\node[draw, circle] (16) at (0,2.8) {};
			
			\path   (1) edge [bend left=0]
			node [above] {} (4);
			\path (2) edge [bend right=0]
			node [above] {} (5);
			
			\path   (4) edge [bend right=0]
			node [above] {} (9);
			\path   (4) edge [bend right=0]
			node [above] {} (10);
			\path   (5) edge [bend left=0]
			node [above] {} (11);
			\path (5) edge [bend right=0]
			node [above] {} (12);
			
			\path   (9) edge [bend right=10]
			node [above] {} (16);
			\path (9) edge [bend right=-10]
			node [above] {} (16);
			\path   (10) edge [bend left=10]
			node [below] {} (15);
			\path   (10) edge [bend left=-10]
			node [below] {} (15);
			\path   (12) edge [bend right=10]
			node [above] {} (16);
			\path (12) edge [bend right=-10]
			node [above] {} (16);
			\path   (11) edge [bend left=10]
			node [below] {} (15);
			\path   (11) edge [bend left=-10]
			node [below] {} (15);
		\end{tikzpicture}\\
	\centering
		\begin{tikzpicture}[scale=0.5, transform shape]
			\node[draw, circle, fill=green] (1) at (-2,0)  {$i$};
			\node[draw, circle, fill=green] (2) at (2,0) {$j$};
			\node[draw, circle, fill=blue!30] (3) at (-2.8,1)  {};
			\node[draw, circle, fill=blue] (4) at (-1.2,1) {};
			\node[draw, circle, fill=blue] (5) at (1.2,1) {};
			\node[draw, circle, fill=blue!30] (6) at (2.8,1) {};
			\node[draw, circle, fill=red] (7) at (-3.2,2)  {};
			\node[draw, circle, fill=red] (8) at (-2.4,2) {};
			\node[draw, circle, fill=red] (9) at (-1.6,2) {};
			\node[draw, circle, fill=red] (10) at (-0.8,2) {};
			\node[draw, circle, fill=red] (11) at (0.8,2)  {};
			\node[draw, circle, fill=red] (12) at (1.6,2) {};
			\node[draw, circle, fill=red] (13) at (2.4,2) {};
			\node[draw, circle, fill=red] (14) at (3.2,2) {};
			
			\node[draw, circle] (15) at (0,2)  {};
			\node[draw, circle] (16) at (0,2.8) {};
			\node[draw, circle] (17) at (0,3.6) {};
			\node[draw, circle] (18) at (0,4.4) {};
			
			\path   (1) edge [bend right=0]
			node [above] {} (2);
			\path   (1) edge [bend right=0]
			node [above] {} (3);
			\path   (1) edge [bend left=0]
			node [above] {} (4);
			\path (2) edge [bend right=0]
			node [above] {} (5);
			\path   (2) edge [bend right=0]
			node [below] {} (6);
			
			\path   (3) edge [bend left=0]
			node [below] {} (7);
			\path  (3) edge [bend left=0]
			node [below] {} (8);
			\path   (4) edge [bend right=0]
			node [above] {} (9);
			\path   (4) edge [bend right=0]
			node [above] {} (10);
			\path   (5) edge [bend left=0]
			node [above] {} (11);
			\path (5) edge [bend right=0]
			node [above] {} (12);
			\path   (6) edge [bend right=0]
			node [below] {} (13);
			\path   (6) edge [bend left=0]
			node [below] {} (14);
			
			\path   (7) edge [bend left=10]
			node [below] {} (18);
			\path  (7) edge [bend left=-10]
			node [below] {} (18);
			\path   (8) edge [bend right=10]
			node [above] {} (17);
			\path   (8) edge [bend right=-10]
			node [above] {} (17);
			\path   (9) edge [bend right=10]
			node [above] {} (16);
			\path (9) edge [bend right=-10]
			node [above] {} (16);
			\path   (10) edge [bend left=10]
			node [below] {} (15);
			\path   (10) edge [bend left=-10]
			node [below] {} (15);
			\path   (14) edge [bend left=10]
			node [below] {} (18);
			\path  (14) edge [bend left=-10]
			node [below] {} (18);
			\path   (13) edge [bend right=10]
			node [above] {} (17);
			\path   (13) edge [bend right=-10]
			node [above] {} (17);
			\path   (12) edge [bend right=10]
			node [above] {} (16);
			\path (12) edge [bend right=-10]
			node [above] {} (16);
			\path   (11) edge [bend left=10]
			node [below] {} (15);
			\path   (11) edge [bend left=-10]
			node [below] {} (15);
		\end{tikzpicture}
	\caption{The network in Figure~\ref{counterex3}(d) is series-parallel. Then, it can be obtained by recursively making parallel and series compositions of series-parallel networks as shown in this figure.}
	\label{sp}
\end{figure}
We prove that the double tree network is not recurrent by showing that $p_i(T_i < T_{\mc{N}_d})$ is the same as in a biased random walk. Indeed, from any $d$ the probability of going from a node at distance $d$ from $i$ to a node at distance $d+1$ and $d-1$ are $2/3$ and $1/3$, respectively. Hence, the double tree is equivalent to a biased random walk on a line as in Figure~\ref{biased_rw}, which is not recurrent \cite[Example 21.2]{levin2017markov}.
Since in the actual network and in the cut network there are no paths between $i$ and $j$ except link $l = \{i,j\}$ (see Figure~\ref{counterex3} (a) and (b)), $r_{l}=r_{l}^{U_d}=1$.
Computing $r_{l}^{L_d}$ is more involved. First, referring to Figure~\ref{counterex3}, we note that, because of the symmetry of the network, the effective resistance between $i$ and $j$ in the shorted network (c), which is $r_{l}^{L_d}$, is equivalent to the effective resistance in (d). Indeed, if we set voltage $v_i=1$ and $v_j=0$, because of symmetry every yellow node has voltage $1/2$. Thus, adding infinite conductance between all of them, i.e., shorting them, does not affect the current in the network (this procedure is also known in literature as \emph{gluing}, see \cite[Section 9.4]{levin2017markov}), and therefore the effective resistance.
The network (d) is series-parallel, so that the effective resistance can be computed iteratively. Specifically, we refer to Figure~\ref{sp} to illustrate the recursion that leads to $r_{l}^{L_d}$. From top to bottom, one can see that the first network has effective resistance between the two blue nodes equal to $3$. The second network is the parallel composition of two of these, in series with two single links. This procedure is iteratively repeated $d-1$ times (in Figure~\ref{sp} only once, since $d=2$), leading to a network that, composed in parallel with a copy of itself and with a single link, is $\mathcal{G}_{l}^{L_d}$. Hence, $r_{l}^{L_d}$ is the result of the following recursion.
\begin{equation*}
	\begin{cases}
		
		r(0)=3, \\
		r(n)=2+\frac{r(n-1)}{2}, \quad d > n \ge1,\\
		r_{l}^{L_d}=(1+\frac{2}{r(d-1)})^{-1},
	\end{cases}
\end{equation*}
which has solution 
\begin{equation*}
	\begin{cases}
		r(n)=(2^{d+2}-1)/2^d, \quad d > n \ge1,\\
		
		r_{l}^{L_d}=\frac{2^{d+1}-1}{2^{d+1}+2^{d}-1} \xrightarrow[d \rightarrow + \infty]{} \frac{2}{3}.
	\end{cases}
\end{equation*}
\end{document}